\documentclass{article}

\usepackage[main, final]{neurips_2026}

\usepackage[utf8]{inputenc} 
\usepackage[T1]{fontenc}    
\usepackage{hyperref}       
\usepackage{url}            
\usepackage{booktabs}       
\usepackage{amsfonts}       
\usepackage{nicefrac}       
\usepackage{microtype}      
\usepackage{xcolor}         

\usepackage{graphicx}
\usepackage{amsmath,amssymb,amsthm}
\usepackage{multirow}
\usepackage{subcaption}
\usepackage{wrapfig}
\usepackage[table]{xcolor}
\usepackage{wrapfig}
\newtheorem{lemma}{Lemma}
\newtheorem{theorem}{Theorem}
\usepackage{pifont}

\title{Mind the Gap No More: Achieving Zero-Gap Multimodal Integration via One Tokenizer}

\author{%
Yanan Li$^{1}$\thanks{Equal contribution. $\dagger$: Corresponding author.}, \hspace{0.5mm} Christina Yi Jin$^{1*}$, Yuan Jin$^{1*}$, Manli Luo$^{1*}$, Tie Xu$^{1*}$, Shuai Jiao$^2$, Wei He$^1$ \\
\And Qing Zhang$^{1\dagger}$\\
$^1$Research Center for Frontier Fundamental Studies, Zhejiang Lab \\
$^2$Research Center for Scientific Data Hub, Zhejiang Lab\\
\texttt{\{liyn,cyj,\hspace{0.5mm}jin.yuan,\hspace{0.5mm}lml0502,fexutie,\hspace{0.5mm}jiaoshuai,\hspace{0.5mm}hewei\}@zhejianglab.org} \\
\texttt{qing.zhang@zhejianglab.org}
}

\begin{document}

\maketitle

\begin{abstract}
A central challenge in developing Multimodal Large Language Models (MLLMs) is effectively integrating heterogeneous inputs into a cohesive reasoning engine. Current paradigms predominantly rely on modular architectures that introduce modality-specific encoders and cross-modal fusion mechanisms. However, these designs are fundamentally bottlenecked by a geometric modality gap, forcing the LLM to expend significant computational capacity on geometric reconciliation rather than deep cross-modal reasoning. In this work, we formally characterize this modality gap and theoretically demonstrate that native architectures, specifically those employing a unified vocabulary, intrinsically maintain a zero-gap state across all hidden layers. 
Guided by these theoretical findings, we propose \textit{One Tokenizer}, a native architecture that maps all modalities directly into a shared token space. We empirically validate this framework on a DNA--text multimodal testbed. Our extensive evaluations reveal that by achieving seamless integration within the LLM's native latent space, One Tokenizer consistently outperforms encoder-based modular counterparts, providing a fundamentally superior framework for deep biological reasoning.

\end{abstract}
\section{Introduction}
\label{sec:intro}

Extending large language models (LLMs) to process heterogeneous modalities, ranging from continuous visual signals to discrete biological sequences, has become a central focus in foundation model research \cite{wang2024qwen2,wang2024emu3,yin2024survey,team2025kimi}. To manage this modality heterogeneity, current multimodal large language models (MLLMs) predominantly rely on a modular architecture. In this standard paradigm, modality-specific encoders (e.g., vision or DNA encoders) are employed to extract heterogeneous representations, which are subsequently mapped into the LLM's shared latent space via cross-modal projectors. While this "plug-and-play" late-fusion approach has facilitated rapid multimodal scaling and demonstrated strong performance on downstream tasks, it inherently postpones deep cross-modal interactions, maintaining a strict representational boundary between modalities prior to their forced alignment. This separation leads to a well-known problem: the \textit{modality gap} \cite{liang2022mind}.

Although the modular design is often viewed as a practical engineering compromise, the resulting modality gap creates a fundamental bottleneck for cross modal reasoning. Specifically, because modality-specific encoders are optimized independently, they naturally embed data into isolated sub-manifolds. 
This creates a fundamental dilemma for modular architectures: heuristically forcing these manifolds to overlap often leads to representation collapse, yet leaving them separated requires the LLM's cross-modal attention to compute similarity across geometrically disjoint spaces, introducing structural biases that obscure fine-grained semantic relevance. To overcome this limitation, a growing line of recent foundation models (e.g., Chameleon \cite{team2024chameleon}, Fuyu \cite{fuyu-8b}, Emu3 \cite{wang2024emu3}) has shifted towards a unified-vocabulary paradigm, abandoning external encoders entirely by discretizing non-textual signals directly into native tokens. Despite its remarkable empirical scalability, the geometric mechanism underlying this architecture remains largely unexplored. Motivated by these observations, in this paper we seek to provide a theoretical understanding of this architectural shift. We argue a simple principle: for LLMs to perform effective cross-modal reasoning, their representations should be gap-free \textit{by design}, not by \textit{forced alignment}.

To formalize this principle, we introduce a statistical framework to analyze the representation dynamics of multimodal architectures. We first examine the embedding initialization (Lemma \ref{lemma:main_init_gap}). We show that modular dual encoders, operating in independent representation spaces, inevitably suffer from a strictly positive initial gap due to high dimensional geometry. In contrast, under the native architecture, all modalities share a unified symbolic vocabulary and a single underlying distribution, driving the expected initial gap to zero as the vocabulary size increases. Furthermore, we demonstrate that this initial geometric state determines the behavior of deep representations (Theorem \ref{thm:main_deep_gap}). Given the Lipschitz continuity of LLM layers, initial geometric distances cannot be easily compressed through depth. Consequently, the initial gap in modular designs is strictly preserved throughout the forward pass, imposing a persistent computational burden as the network struggles to force alignment.  By contrast, our native gap-free initialization guarantees an intrinsically integrated representation space at all depth, eliminating the need for post hoc geometric alignment. 

Empirically validating this claim is challenging because tokenizing continuous data (like images) introduces noise that blurs the architectural comparison. While our gap-free guaranty applies to any modality, observing it in practice requires a clean testbed. We therefore focus on inherently discrete modalities: DNA and text. This allows us to bypass complex external encoders and represent DNA $k$-mers directly as native LLM tokens. Our results on DNA-text benchmarks are clear: using the same foundation model and training data, the native architecture consistently outperforms modular baselines, even when the latter are enhanced with CLIP-style contrastive alignment. This suggests that for discrete sequences, native integration is a simpler and more powerful route than building complex alignment modules. We summarize our main contributions as follows:
\begin{itemize}
\item \textbf{Theoretical guarantee.} We mathematically prove that our native architecture achieves and strictly preserves a zero gap state across layers, fundamentally eliminating the geometric bottleneck that plagues modular designs.
\item \textbf{Controlled evaluation.} We establish a noise free testbed using inherently discrete modalities (DNA and text) to isolate an architecture's reasoning capacity from the tokenization artifacts of continuous data.
\item \textbf{Empirical superiority.} On DNA-text benchmarks, our native integration consistently outperforms modular baselines, even those enhanced with complex CLIP-style contrastive alignment.
\end{itemize}

 \section{Related Work}
 \label{sec:related}

 \textbf{Modular Multimodal Large Language Models.} 
 The success of large language models (LLMs) has driven recent efforts to integrate multiple modalities. Originally developed for vision-language tasks \cite{radford2021learning, liu2023visual, wang2024qwen2, team2025kimi}, these modular architectures have naturally extended to the biological domain for processing genomic and molecular sequences. Most existing methods use a pre-trained modality-specific encoder connected to an LLM via projection modules \cite{dai2023instructblip, liu2024improved}. In DNA-text modeling, established baselines like BioReason \cite{fallahpour2025bioreason} and ChatNT \cite{de2025multimodal} align genomic features with text using explicit projection layers. However, as pointed out by \cite{diao2024unveiling}, such encoder-based architectures suffer from inherent practical and structural limitations, including strict reliance on pre-trained encoders, deployment inefficiencies, and early-stage geometric misalignments that force the LLM to expend substantial computational capacity on cross-modal alignment.

\textbf{Native Multimodal Large Language Models.} 
To address the bottlenecks of modular projections, recent research has shifted toward encoder-free, native MLLMs. This paradigm typically processes all modalities within a single autoregressive framework, representing multi-modal inputs as token sequences that are directly integrated into a unified vocabulary \cite{team2024chameleon, xie2024show}. By eliminating modality-specific encoders, native architectures enable cross-modal interaction from the earliest layers. This design principle has recently started to emerge in biological sequence modeling. For instance, Omni-DNA extends the native modeling paradigm to DNA-text settings by building a DNA-centric vocabulary augmented with a limited set of natural language tokens \cite{li2025omnidna}. However, while these native architectures demonstrate strong empirical performance, existing studies largely attribute their success to scaling laws, data composition, or engineering heuristics \cite{team2024chameleon, xie2024show, li2025omnidna}. The fundamental geometric mechanism underlying this advantage, specifically how native fusion closes the modality gap, remains largely unexplored.

\textbf{Modality Gap and Latent Geometry.} 
A fundamental challenge in multimodal learning is the modality gap: the geometric separation between representations of different modalities in a shared latent space. Prior work demonstrates that multimodal contrastive models embed different modalities into disjoint subspaces, exhibiting a cone-effect separation caused by heterogeneous data distributions \cite{liang2022mind}. Building on this geometric perspective, recent studies explore how strong cross-modal generalization occurs despite this gap, particularly when using frozen LLMs \cite{palit2023towards, schwettmann2023multimodal, shukor2024implicit, parekh2024concept, dang2024explainable}. For instance, the Implicit Multimodal Alignment (IMA) phenomenon suggests that while non-textual tokens remain separated from text in the embedding space, they can still trigger similar computational substructures inside the network \cite{shukor2024implicit}. 
Rather than investigating how models compensate for this persistent separation, our work introduces a latent geometry framework to rigorously quantify how the modality gap evolves layer-by-layer under controlled architectural choices. We theoretically and empirically demonstrate that native embedding via a unified vocabulary intrinsically closes this gap, offering a fundamental structural advantage over modular designs.

\section{Methodology}
\label{sec:proposed_method}
In this section, we formalize the theoretical foundations of our approach. We first establish a general formulation for MLLM and mathematically define the modality gap. Based on this framework, we demonstrate our core theoretical finding: while modular designs inherently suffer from a modality gap at the embedding layer, native unified-vocabulary MLLMs naturally eliminate this penalty and maintain an integrated representation space across all deep layers.

\textbf{General framework for analyzing MLLMs.}
Let $H^m = \{h^m_1, h_2^m, \ldots, h^m_{N_m}\}$ and $H^t = \{h^t_1, h^t_2, \ldots, h^t_{N_t}\}$ denote the continuous representation (embeddings) mapped into the LLM's shared latent space for modality $m$ and text $t$, respectively. Let $O = \{o_1, o_2, \ldots, o_{N_o}\}$ denote the corresponding output sequence. The standard next-token prediction (NTP) objective is formulated as: 
\begin{equation}
    \mathcal{L}_{\text{NTP}} = -\sum_{k=1}^{N_o} \log p(o_k|H^m, H^t, o_{<k}; \theta)
\label{eq:next_token_prediction_objective}
\end{equation}
where $\theta$ denotes the learnable parameters. Under this objective, we compare two dominant MLLM paradigms:
\begin{itemize}
    \item \textbf{Modular architecture:} combines an independently pretrained modality-specific encoder (denoted as $E_m$) with an LLM using a projection module $C$. Specifically, given the raw modality input $\mathcal{X}_m$ and tokenized text $\mathcal{X}_t$, the embeddings are obtained via $H^m = C(E_m(\mathcal{X}_m))$ and $H^t = E(\mathcal{X}_t)$, where $E$ is the LLM's text embedding layer. The output is expressed as:
    \begin{equation}
    O = LLM([C(E_m(\mathcal{X}_m)); E(\mathcal{X}_t)])=LLM([H^m; H^t])
    \end{equation}

    \item \textbf{Native unified-vocabulary architecture:} integrates multiple modalities within a shared token space and backbone. Given the tokenized modality input $\mathcal{X}_m$ and tokenized text $\mathcal{X}_t$, both representations are obtained directly from the same embedding layer $E$: $H^m = E(\mathcal{X}_m)$ and $H^t = E(\mathcal{X}_t)$. The output is simply:
    \begin{equation}
    O = LLM([E(\mathcal{X}_m); E(\mathcal{X}_t)])=LLM([H^m; H^t])
    \end{equation}
\end{itemize}

\subsection{Theoretical Analysis}
\label{sec:theory}
We now establish a comparative mathematical analysis of the representational differences between native MLLMs and their modular counterparts. While prior work \cite{liang2022mind} formalized the modality gap simply as the absolute Euclidean distance between empirical centroids (i.e., $\| \bar{\mathbf{x}} - \bar{\mathbf{y}} \|_2$), this absolute metric is highly scale-dependent. When comparing fundamentally different architectures, the absolute magnitude of embedding vectors can vary significantly, making direct distance comparisons skewed. 

To ensure a rigorous and scale-invariant comparison, we elevate this concept into a statistical framework. Inspired by Cohen's $d$ \cite{cohen2013statistical}, we redefine the modality gap by normalizing the centroid distance with the intra-modality variance. Throughout our analysis, we use the superscript $(l)$ to denote the representations at the $l$-th layer of the LLM, where $l=0$ specifically refers to the initial continuous representations mapped by the embedding layer.

\textbf{Definition 1} (Modality Gap). 
\label{def:main_def} Let $\mathcal{H}_A$ and $\mathcal{H}_B$ be the latent sub-manifolds for any two distinct modalities $A$ and $B$, respectively. At any given layer $l$ of the LLM, let $c_A^{(l)}$ and $c_B^{(l)}$ denote the expected representations (centroids) of $\mathcal{H}_A$ and $\mathcal{H}_B$, and let $(\sigma_A^{(l)})^2$ and $(\sigma_B^{(l)})^2$ denote their respective average intra-modality variances. The modality gap $\Delta^{(l)}(\mathcal{H}_A, \mathcal{H}_B)$ is defined as the normalized Euclidean distance between the expected representations of the two modalities:
\begin{equation}
\Delta^{(l)}(\mathcal{H}_A, \mathcal{H}_B) = \frac{\|c_A^{(l)} - c_B^{(l)}\|_2}{\sqrt{\frac{(\sigma_A^{(l)})^2 + (\sigma_B^{(l)})^2}{2}}}
\end{equation}
This normalized formulation offers a crucial theoretical advantage: it isolates the structural separation of modalities from trivial scaling artifacts. Furthermore, because modern LLMs inherently constrain variance parameters ($\sigma^2$) to be strictly positive and bounded (e.g., via initialization and normalization layers), our normalized metric $\Delta$ shares a strict monotonic relationship with the absolute gap defined by \cite{liang2022mind}. Consequently, theoretically minimizing our normalized gap rigorously guarantees the closure of the modality gap in its original sense. Therefore, in the remainder of this paper, we use the term \textit{modality gap} to refer to this normalized metric $\Delta$, unless specified otherwise.

Based on this robust definition, we present the following lemma showing that at the input layer, the modular architecture inherently suffers from a larger modality gap than its native counterpart.

\begin{lemma}[Embedding-level gap]
\label{lemma:main_init_gap}
\textit{Let $V_A$ and $V_B$ denote the modality-specific vocabularies for modalities $A$ and $B$, respectively. The expected modality gaps at the initial embedding layer ($l = 0$) of the modular architecture $\Delta_{\text{Mod}}^{(0)}(\mathcal{H}_A, \mathcal{H}_B)$ and the native architecture $\Delta_{\text{Nat}}^{(0)}(\mathcal{H}_A, \mathcal{H}_B)$ satisfy:}
$$\mathbb{E}[\Delta_{\text{Mod}}^{(0)}(\mathcal{H}_A, \mathcal{H}_B)] > \lim_{|V_A|, |V_B| \to \infty} \mathbb{E}[\Delta_{\text{Nat}}^{(0)}(\mathcal{H}_A, \mathcal{H}_B)] = 0 $$
\end{lemma}
\textit{Proof.} Refer to the Appendix \ref{sec:proof}. 

While Lemma \ref{lemma:main_init_gap} establishes a structural advantage at the initial embedding layer, modern LLMs are characterized by massive, deep non-linear transformations. A critical question arises: does this geometric advantage survive throughout the deep forward pass, or does the modality gap eventually reopen? Theorem \ref{thm:main_deep_gap} demonstrates that, under standard continuity assumptions, the native architecture inherently preserves this minimal gap across all hidden layers.
\begin{theorem}[Layer-wise gap preservation]
\label{thm:main_deep_gap}
Given the initial condition established in Lemma \ref{lemma:main_init_gap}, and assuming the layer-wise transformation functions are Lipschitz continuous, the expected modality gaps $\Delta_{\text{Mod}}^{(l)}$ and $\Delta_{\text{Nat}}^{(l)}$ at every subsequent layer $l$ ($l > 0$) throughout the forward pass satisfy:
\[
   \mathbb{E}[\Delta_{\text{Mod}}^{(l)}(\mathcal{H}_A, \mathcal{H}_B)] > \lim_{|V_A|, |V_B| \to \infty} \mathbb{E}[\Delta_{\text{Nat}}^{(l)}(\mathcal{H}_A, \mathcal{H}_B)] = 0
\]
\end{theorem}
\textit{Proof.} Refer to the Appendix \ref{sec:proof}. 

Theorem 1 bridges our theoretical analysis with practical architecture design. It guarantees that the native paradigm maintains a minimal modality gap $\Delta^{(l)}_{\text{Nat}}$ across all Transformer layers. This mathematical property reveals a fundamental architectural divergence: because modular baselines preserve their initial geometric separation throughout the forward pass, they must rely on explicit projection modules to forcefully bridge the modalities. Consequently, this persistent separation forces the LLM to expend its computational capacity on geometric reconciliation rather than deep semantic reasoning. In contrast, by preserving an inherently unified representation space, the native architecture achieves seamless cross-modal integration without requiring explicit alignment interventions.

\begin{figure}[!htp]
    \centering
    \includegraphics[width=0.9\linewidth]{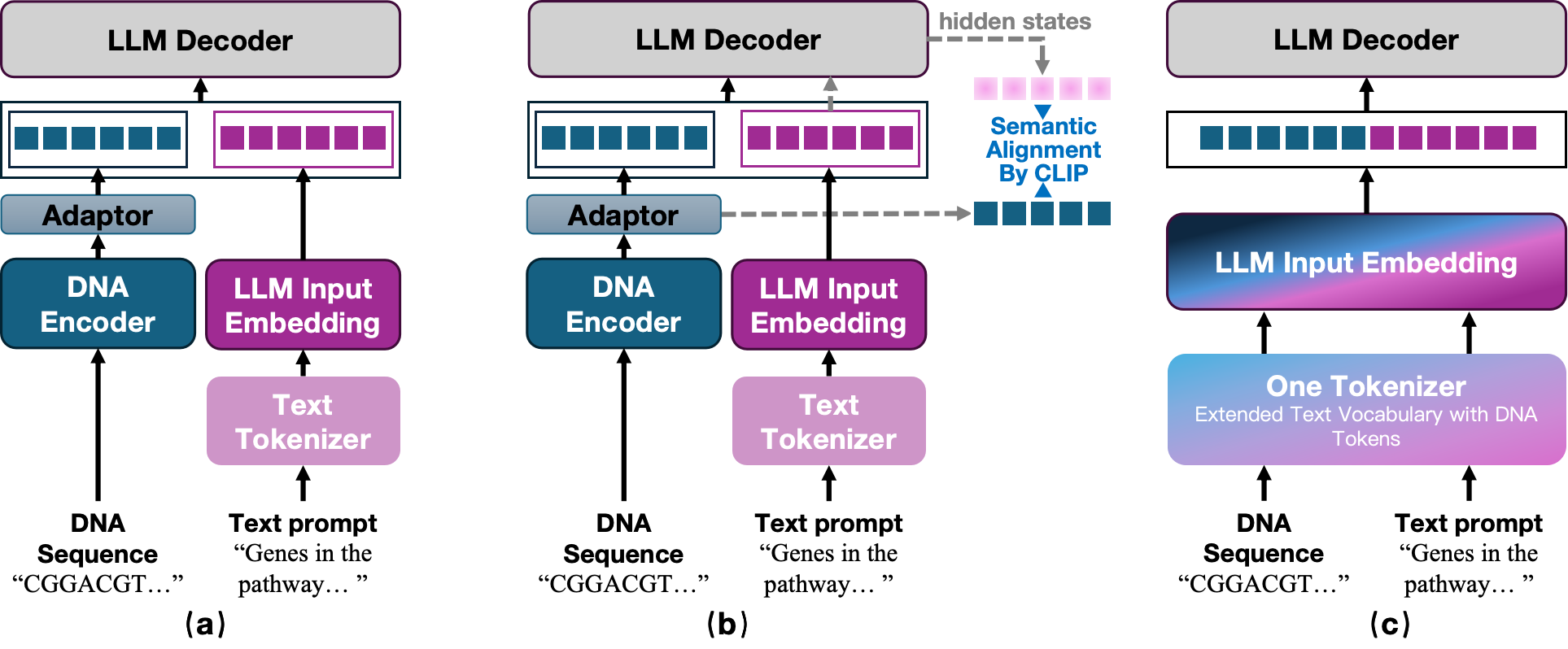}
    \caption{We systematically investigate three DNA-text fusion strategies. (a) The standard modular architecture, adopted by current DNA-text models. (b) SeqCLIP: explict semantic alignment on the gene encoder by contrastively learning on massive DNA-text pairs. (c) Our One Tokenizer: extend the pre-trained LLM's vocabulary with DNA-specific tokens, allowing LLM to process them natively. }
    \label{fig:dna-llm-frameworks}
\end{figure}

\subsection{Empirical Textbed: DNA-Text Modeling}
As motivated in Section \ref{sec:intro}, we utilize DNA-text modeling as a controlled, discrete multimodal environment. This allows us to empirically validate our theoretical findings without the confounding artifacts introduced by continuous data tokenization such as vector quantization (e.g, MoVQGAN \cite{zheng2022movq}).

\textbf{Task formulation.} We formulate DNA-text modeling as a standard cross-modal conditional generation task. Formally, let $S$ denote a DNA sequence, and $Q$ denote a tokenized English query (e.g., obtained via Byte-Pair Encoding \cite{sennrich2016neural}). Let $A = (a_1, a_2, \dots, a_{N_a})$ denote the corresponding English answer sequence of length $N_a$. Our goal is to learn a generic DNA-text model $f_\theta$, parameterized by $\theta$, which outputs the token-level conditional probability distribution $p$ to autoregressively generate the answer:
$
p(A|S, Q; \theta) = \prod_{k=1}^{N_a} p(a_k|S, Q, a_{<k}; \theta).
$

\textbf{Modular baselines and SeqCLIP.} Current DNA-text models (Fig. \ref{fig:dna-llm-frameworks}(a)) typically adopt a modular design, relying on different cross-modal adapters. For instance, BioReason \cite{fallahpour2025bioreason} projects per-token embeddings via a learnable linear layer, whereas ChatNT \cite{de2025multimodal} aggressively compresses genomic tokens into 64 query embeddings using a nonlinear Q-former. 
To facilitate a fair and representative comparison with state-of-the-art multimodal alignment approaches, we introduce an enhanced baseline: \textbf{SeqCLIP} (Fig. \ref{fig:dna-llm-frameworks}(b)). 
Motivated by the widespread success of contrastive pre-training in vision-language foundation models, we apply an InfoNCE contrastive loss between the projected genomic representations and the corresponding text representations using paired gene-text data \cite{dhanasekar2025genechat} during pre-training. Rather than relying solely on the generative loss, this contrastive objective explicitly forces the adapter to extract semantically meaningful and aligned representations before the autoregressive generation phase.

\textbf{Native architecture: One Tokenizer.} To bypass the inherent modality gap of modular designs, we adopt a native integration.  Although Omni-DNA \cite{li2025omnidna} also attempts vocabulary mixing, its proprietary design remains fundamentally \textit{DNA-centric}. Instead of a balanced integration, it merely appends a strictly limited set of text tokens (around $4k$) to a genomic tokenizer. This approach severely restricts the broad linguistic competence and complex reasoning capabilities available in modern LLMs. Furthermore, its closed-source nature precludes direct empirical comparison.

Instead, we propose an open, \textit{LLM-centric} native integration: \textbf{One Tokenizer} (Fig. \ref{fig:dna-llm-frameworks}(c)). Let $V_{\text{text}}$ denote the LLM's rich textual vocabulary and $V_{\text{DNA}}$ denote the discrete DNA $k$-mer vocabulary. We construct a unified representation space $V = V_{\text{text}} \cup V_{\text{DNA}}$. Consequently, the original LLM embedding layer is expanded to utilize a shared parameter matrix $\mathbf{E} \in \mathbb{R}^{|V| \times d}$. Let $\text{Emb}(\cdot)$ denote the embedding lookup function parameterized by $\mathbf{E}$. Both DNA and text tokens are mapped directly into the shared latent space through this single layer. The concatenated input sequence $[\text{Emb}(S); \text{Emb}(Q)]$ is then processed uniformly by the LLM backbone $f_\theta$.This homogeneous architecture requires zero modality-specific encoders or projection parameters, naturally executing the zero-gap initialization and preserving deep semantic integration. 

Furthermore, this vocabulary expansion poses minimal risk of representation drift caused by vocabulary distribution shifts \cite{li2025omnidna}. Given that the native text vocabulary $V_{\text{text}}$ in modern LLMs (e.g., Qwen) encompasses over 150,000 tokens, the addition of a DNA vocabulary (e.g., $|V_{\text{DNA}}| = 15,625$ for $k=6$) constitutes only a marginal expansion of roughly 10\%. This structural asymmetry ensures that the LLM's core language representations remain highly stable during fine-tuning, intrinsically preserving its reasoning capabilities. For comparison, the DNA-centric model Omni-DNA employs a roughly 1:1 token ratio (4,096 DNA tokens versus 4,077 text wordpieces). Because their text and genomic token spaces are comparable in scale, they are highly susceptible to this distribution shift and subsequent representation drift, forcing them to introduce explicit methodological workarounds to address the instability.

\section{Experiments}
\label{sec:experiments}
To validate our theoretical claims and evaluate the proposed architecture, we first probe the latent space to verify whether the observed modality gap aligns with the theoretical bounds established in Lemma \ref{lemma:main_init_gap} and Theorem \ref{lemma:main_init_gap} (Sec.\ref{subsec:theoretical_validation}). We then investigate how the reduction of this gap translates to superior downstream performance against modular baselines (Sec.\ref{subsec:downstream_performance}). Finally, we explore the optimal configuration choices, specifically regarding embedding tuning strategies and $k$-mer granularity, to maximize the efficacy of the One Tokenizer (Sec.\ref{subsec:ablation}).

\subsection{Experimental Settings}
\textbf{Datasets.} We evaluate our method on two widely adopted benchmarks targeting different downstream capabilities: the Nucleotide Transformer Benchmark (NT) \cite{de2025multimodal} for classification, and the KEGG-Derived Biological Reasoning Dataset (KEGG) \cite{fallahpour2025bioreason} for complex reasoning. The NT benchmark consists of relatively simple question-answer pairs, where each instance pairs a natural language question with one or two DNA sequences. In contrast, the KEGG benchmark emphasizes multi-step biological deduction over simple multi-choice tasks, requiring models to synthesize biological knowledge before generating the final response. To ensure fair comparison, we adopt the standard training, validation, and test splits established in prior works \cite{fallahpour2025bioreason, de2025multimodal}. Additional dataset details are deferred to the Appendix \ref{appendix:datasets} due to space constraints.

\textbf{Evaluation metrics.}
We employ task-specific evaluation metrics tailored to the requirements of classification and biological reasoning:

\textit{Classification.} For the NT benchmark, we report the Matthews Correlation Coefficient (MCC) across all 18 tasks, maintaining consistency with established evaluation protocols  \cite{schiff2024caduceus, de2025multimodal}. For the KEGG benchmark, we follow the assessment framework introduced in \cite{fallahpour2025bioreason}, reporting Accuracy, macro F1-score, precision, and recall. These metrics collectively ensure a robust evaluation of the models' ability to categorize biological sequences and predict their functional annotations across diverse tasks.

\textit{Reasoning.} We evaluate the reasoning performance using \textit{reasoning score}, calculated through a rank-based LLM-as-a-judge protocol \cite{chen2024mllm}. A panel of 3 LLMs (GPT-5.1, Gemini-3.1-Pro-Preview, and o3) jointly ranks the outputs of all four evaluated models alongside the ground-truth reference. Rankings are converted into ordinal scores (from $5$ for the best to $1$ for the worst). For each sample, the scores from the 3 LLM judges are summed, and the final performance is computed by averaging these accumulated scores across all test samples. Importantly, by including the ground truth as a candidate in the ranking pool, it receives an aggregated score that serves as the empirical upper bound. This design allows us to not only rank the models but also directly quantify the performance gap between each model's reasoning and the ideal ground truth. Further details are deferred to the Appendix \ref{appendix:llm}.

\textbf{Implementation details.}
We employ Qwen3-1.7B \cite{yang2025qwen3} as the language model backbone and a $6$-mer DNA tokenizer for sequence encoding. All experiments are conducted on $8$ NVIDIA A100 GPUs using the Adam optimizer ($\beta_1 = 0.9$, $\beta_2 = 0.999$, weight decay 0.05, and learning rate $10^{-6}$). Regarding the baselines, we evaluate the modular architectures of BioReason and ChatNT. To ensure a rigorously controlled comparison, we align all underlying LLM backbones to a standardized Qwen3-1.7B, guaranteeing that any observed performance dynamics are directly attributable to the structural differences between modular and native architectures, rather than differences in raw parameter counts. While BioReason natively operates at this capacity, we specifically adapt ChatNT from its default 8B configuration to match this scale. Additionally, we implement SeqCLIP as a comparative baseline, which is fine-tuned using LoRA with a rank of 8. In contrast, we perform full-parameter optimization for our One Tokenizer. All models are trained with an effective batch size of 8 (achieved via a per-device batch size of 1 with 8 gradient accumulation steps).

\begin{figure}[!htp]
    \centering
    \includegraphics[width=\linewidth]{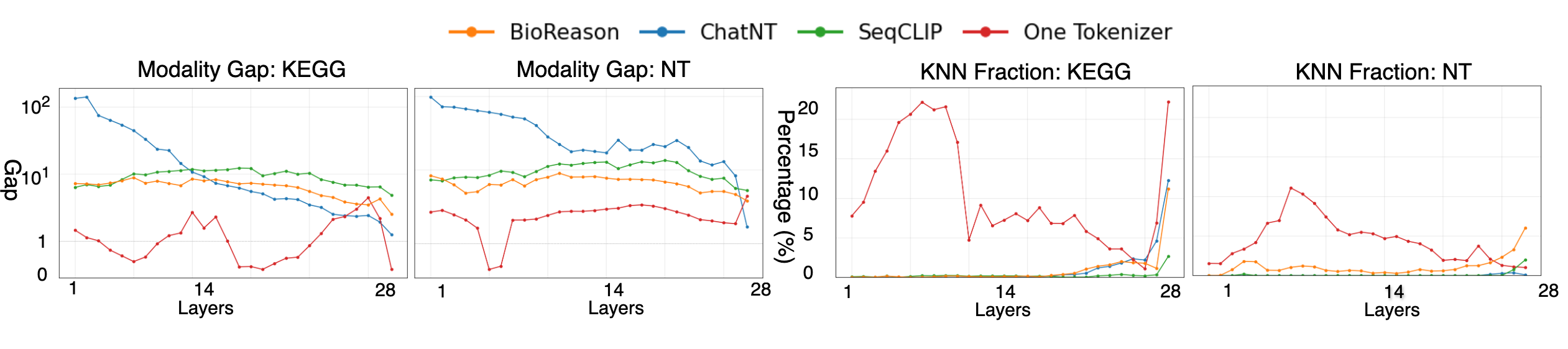}
    \caption{Empirical analysis of the modality gap and KNN fraction on KEGG and NT.  We extract the latent embeddings at each layer and compute both metrics.}
    \label{fig:modality_gap_analysis}
\end{figure}

\begin{figure}[!htp]
    \centering
    \includegraphics[width=0.9\linewidth]{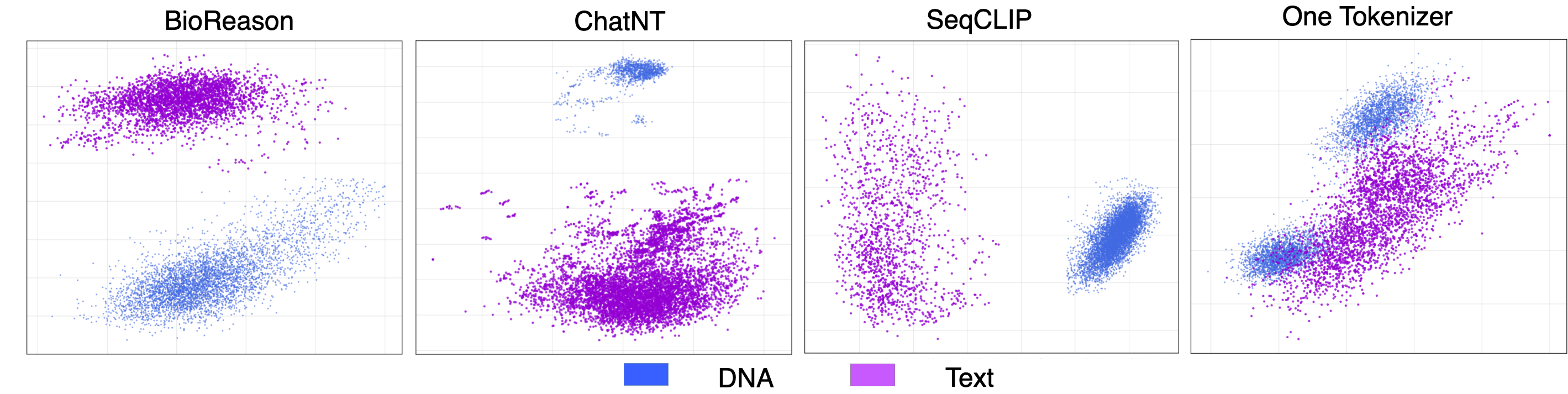}
    \caption{PCA visualization of the 14th-layer embeddings on the KEGG benchmark. One Tokenizer exhibits nearly complete overlap between modalities, whereas others show clear modality gaps. Similar results on the NT benchmark are provided in Appendix \ref{appendix:vis_nt}.}
    \label{fig:visualization_latent_embeddings_kegg}
    \vspace{-5pt}
\end{figure}

\subsection{Theoretical Validation}
\label{subsec:theoretical_validation}
To empirically validate the theoretical bounds established in Section \ref{sec:theory}, we analyze the modality gap across hidden layers using two key metrics: (1) The normalized Euclidean distance formulated in Definition 1, and (2) The cross-modal $k$-nearest neighbor (KNN) fraction ($k=100$) to quantify local representation mixing. A higher fraction indicates that DNA and text tokens are more uniformly interwined in the latent space, reflecting deeper cross-modal integration (details of these two metrics are deferred to the Appendix \ref{appendix:metrics}).

As illustrated in Fig \ref{fig:modality_gap_analysis} and \ref{fig:visualization_latent_embeddings_kegg}, our results across both KEGG and NT benchmarks yield several key insights:
First, on both KEGG and NT, modular architectures exhibit a substantial modality gap at the initial layer ($l = 0$), reflecting the inherent misalignment of independently learned embedding spaces. As evidenced by the near-zero KNN fraction across the first 20 layers, this initial separation forces the model to expend significant computational depth purely on geometric reconciliation, i.e. aligning disparate subspaces and delaying actual cross-modal reasoning to later stages. In contrast, One Tokenizer starts with a significantly reduced initial gap, empirically validating the theoretical advantages established in Lemma \ref{lemma:main_init_gap}. This native alignment eliminates the need for early geometric reconciliation, yielding an initial KNN fraction orders of magnitude larger than its modular competitors.

Second, as representations propagate through deeper layers ($l>0$), modular methods attempt to bridge the modality gap through cross-modal attention, yet a noticeable separation persists until the final layer. On the contrary, One Tokenizer consistently maintains a seamless integration throughout the Transformer layers. This results aligns with Theorem \ref{thm:main_deep_gap}, confirming that a single, unified embedding matrix allows DNA and text to naturally align on a continuous manifold. Notably, Fig. \ref{fig:visualization_latent_embeddings_kegg} visualizes this effect at the 14th layer -- the exact midpoint of the network ($L=28$) where deep semantic integration is expected to stabilize before task-specific routing in the final layers \cite{tenney2019bert}. At this critical computational juncture, One Tokenizer already achieves a nearly seamless overlap of DNA and text embeddings, whereas baseline models maintain distinct, isolated clusters (consistent geometric behaviors are also observed on the NT benchmark, see Appendix \ref{appendix:vis_nt}).

Finally, this native geometric unity fundamentally alters how the LLM approaches multimodal tasks. Because One Tokenizer successfully maps DNA and text onto a truly unified continuous manifold, the LLM no longer processes them as disparate “foreign” entities requiring constant translation. Instead, cross-modal reasoning is simplified into standard autoregressive token prediction over a shared semantic space. This geometric unity intrinsically extends the pre-trained LLM's powerful zero-shot and logical reasoning capabilities to biological sequences, serving as the core architectural driver for the substantial performance gains observed in downstream tasks, as shown in the next section.

\vspace{-10pt}
\begin{table}[!htp]
    \centering
    \caption{Classification performance on KEGG and NT. One Tokenizer achieves the best performance across all metrics.}
    \begin{tabular}{|c |c c c c |c|}
        \hline
        \multirow{2}{*}{\textbf{Method}} & \multicolumn{4}{c|}{\textbf{KEGG}} & \textbf{NT} \\
        \cline{2-6}
        & Accuracy & F1-score & Precision & Recall & MCC \\
        \hline
        BioReason & 88.42 & 72.13 & 75.42 & 71.91 & 0.54\\
        ChatNT & 46.94& 17.58 &21.42 & 15.85 & 0.55\\
         SeqCLIP &  87.50 & 90.31 & 94.44 & 87.50 & 0.56\\
        \textbf{One Tokenizer} &\textbf{99.31}& \textbf{95.24}& \textbf{96.43} &\textbf{94.64} & \textbf{0.62}\\
        \hline
    \end{tabular}
    \vspace{-5mm}
    \label{tab:sota_vep_kegg}
\end{table}

\begin{wrapfigure}{l}{0.5\textwidth}
    \centering
    \includegraphics[width=0.5\textwidth]{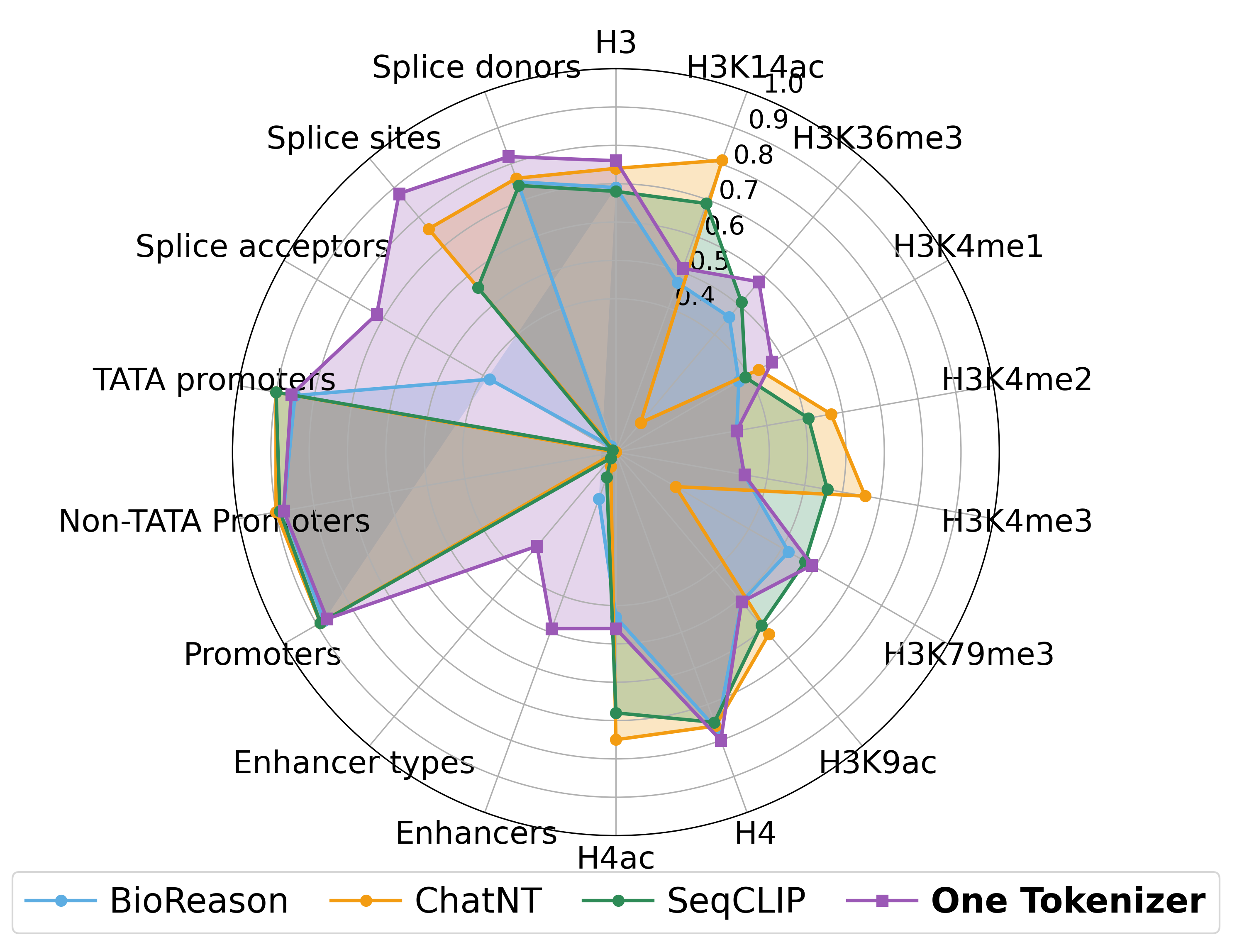}
    \caption{MCC on the 18 tasks in NT.}
    \label{fig:sota_nt}
    \vspace{-3mm}
\end{wrapfigure}

\subsection{Downstream Performance}
\label{subsec:downstream_performance}
\textbf{Classification task. } Performance on the 18 NT tasks is summarized in Fig.~\ref{fig:sota_nt}, while the KEGG classification results are presented in  Tab.~\ref{tab:sota_vep_kegg} (see Appendix \ref{appendix:sota_nt} for details). 
Overall, we observe: (1) SeqCLIP consistently outperforms BioReason, demonstrating that explicit CLIP-style semantic alignment improves adapter-based models. 

(2) One Tokenizer achieves the highest average MCC score ($0.62$) across all 18 NT tasks, with \textit{significant improvements} on both splice- and enhancer- related predictions. These tasks require the model to capture both fine-grained local features and broad global context. Modular architectures typically struggle here because their modality-specific projection layers act as an information bottleneck, inevitably distorting or subsampling the original sequence features. In contrast, One Tokenizer avoids information loss by mapping sequences directly into the LLM’s native space, ensuring the model retains full sequence fidelity for pattern recognition.

(3) On the more complex KEGG benchmark with 37 disease classes, One Tokenizer achieves a SOTA accuracy of $99.31\%$. While modular baselines (BioReason and SeqCLIP) exhibit performance inconsistencies between overall accuracy and macro-F1, One Tokenizer excels in both metrics simultaneously.
By natively embedding genomic tokens within the LLM's semantic space, our architecture establishes more precise decision boundaries for multi-class classification. This native integration effectively resolves the alignment-induced instability inherent in dual-encoder designs, ensuring robust performance even across diverse and imbalanced task distributions.

\begin{wrapfigure}{r}{0.43\textwidth}
    \centering
    \includegraphics[width=0.43\textwidth]{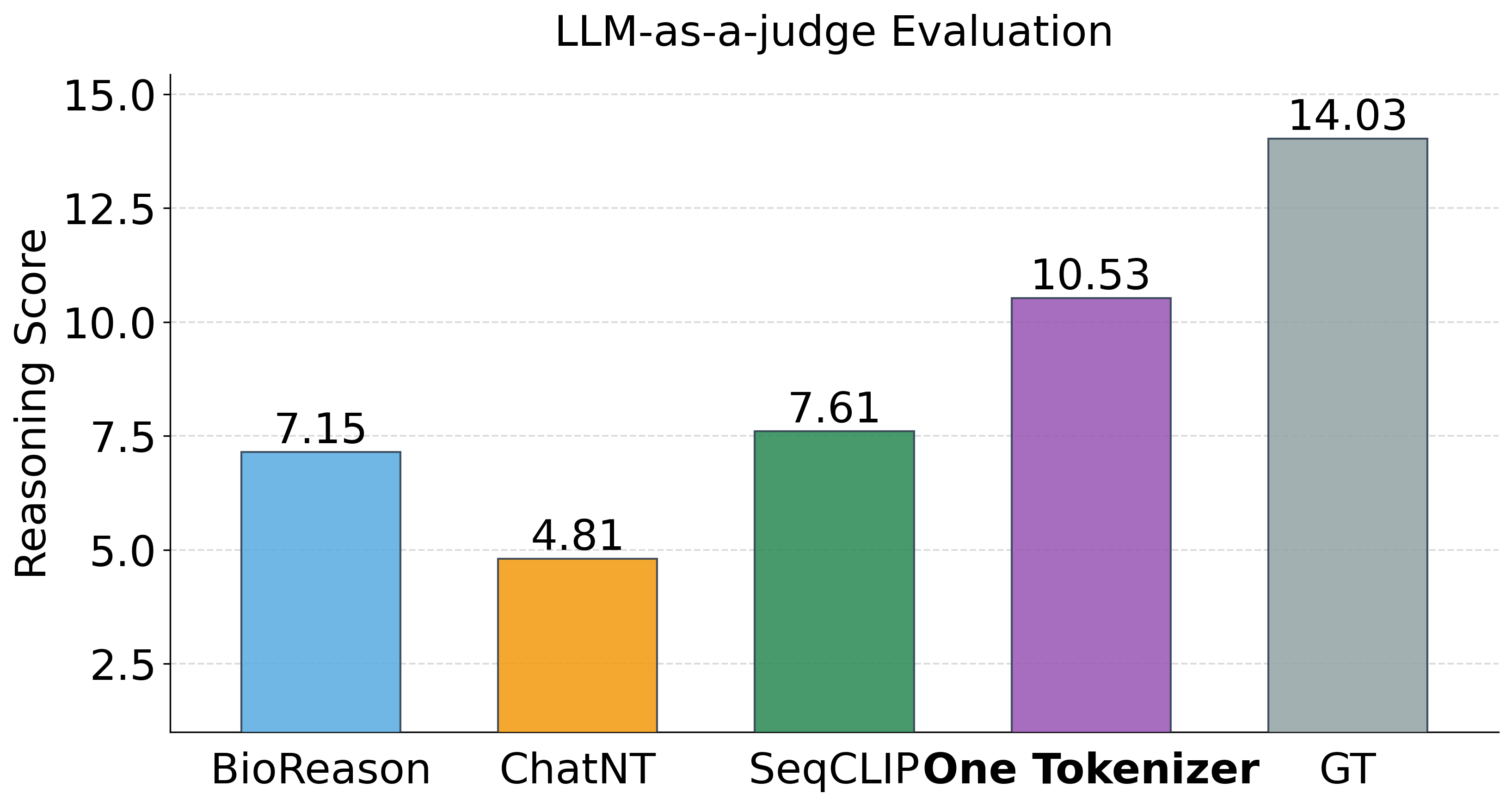}
    \caption{Reasoning score on KEGG, assessed using our LLM-as-a-judge protocol. }
    \label{fig:sota_reasoning}
    \vspace{-5mm}
\end{wrapfigure}

\textbf{Reasoning task.} We evaluate deep cross-modal reasoning using our proposed LLM-as-a-judge protocol (Fig.\ref{fig:sota_reasoning}). 
One Tokenizer achieves a superior reasoning score of 10.53, significantly narrowing the gap to the ground-truth (GT) upper bound of 14.03. While modular baselines (ChatNT, BioReason, and SeqCLIP) lag considerably behind, our native integration successfully unlocks the LLM's inherent reasoning engine for biological sequences. This high-fidelity alignment  facilitates complex logical deductions that are otherwise bottlenecked by traditional projection based designs. Additional case studies are provided in Appendix \ref{appendix:reasoning_results}.

\subsection{Ablation Studies: Constructing and Optimizing One Tokenizer}
\label{subsec:ablation}
To fully realize the potential of One Tokenizer, two critical architectural questions must be addressed: (1) What is the optimal sequence granularity for constructing the DNA tokens? and (2) Once integrated, how should the embedding space be optimized during training? We investigate these through systematic ablation studies on both the NT and KEGG benchmark, as detailed below.

\begin{wraptable}{r}{0.55\linewidth}
\vspace{-12pt}
\centering
\caption{Impact of $k$ in $k$-mer tokenization of DNA sequences on NT classification task. ``Voc Size'', ``Seq'' and ``tt'' are short for vocabulary size, sequence length, and training time, respectively.}
\resizebox{0.95\linewidth}{!}{
\begin{tabular}{ c l c c c c c c}
        \hline
        $k$  & Voc Size & Seq & Acc & F1 & MCC & tt (s) \\
        \hline
        7 & 78,125 & 200 & 0.75 & 0.74 & 0.53 & 1688 \\
        6 & 15,625 &200 &  \textbf{0.80} & \textbf{0.78} & \textbf{0.62} & \textbf{521}\\
       4 & 625 & 280 & 0.80 & 0.78 & 0.61 & 725 \\
       3  & 125 & 330 & 0.80 & 0.79 & 0.62 & 841 \\
       2 & 25 & 420 &  0.79& 0.77 & 0.61& 1074  \\
       1 & 5 & 720 & 0.76 & 0.73 & 0.54 & 4597 \\
        \hline
\end{tabular}
}
\label{tab:ablation_vocab_size}
\vspace{-12pt}
\end{wraptable}

\textbf{Effect of sequence granularity ($k$-mer size).} 
We evaluate the impact of sequence granularity by varying the $k$-mer size from 1 to 7 (Tab. \ref{tab:ablation_vocab_size}), identifying $k=6$ as the optimal setting. Note that we do not explore $k \ge 8$ because the vocabulary size grows exponentially ($5^k$), which inevitably leads to severe data sparsity and disrupts the LLM's pre-trained token distribution.

\textbf{Tuning the shared embeddings.} 
We evaluate different embedding update strategies for the shared embedding layer while keeping the LLM backbone fully fine-tuned (Tab. \ref{tab:ablation_finetune}). Although updating either DNA or text embeddings individually yields functional results, jointly tuning both achieves the best performance across all metrics. This confirms that full optimization of the shared embedding space is essential for seamless cross-modal integration.

\begin{table*}[!htbp]
    \centering
    \caption{Impact of embedding fine-tuning strategies on KEGG. \ding{55} and \ding{51} denote frozen and fine-tuning strategies, respectively. }
    \resizebox{0.9\textwidth}{!}{
    \begin{tabular}{|c c |c c c c |c |}
    \hline
    \multirow{2}{*}{DNA Embeds} & \multirow{2}{*}{Text Embeds} & \multicolumn{4}{c|}{Classification} & \multirow{2}{*}{Reasoning Score} \\
    \cline{3-6}
     & & Acc & F1 & Precision & Recall & \\
    \hline
    \ding{51} & \ding{55} &97.92 & 92.67 & 92.86 & 92.50 & 8.55  \\
    \ding{55} & \ding{51} &  86.81 & 75.70 & 74.85 & 77.68 & 7.14 \\
    \ding{51} & \ding{51} & \textbf{99.31 }& \textbf{95.24} & \textbf{96.43} & \textbf{94.64} & \textbf{10.53} \\
    \hline
    \end{tabular}
    }
    \label{tab:ablation_finetune}
    \vspace{-7pt}
\end{table*}

\section{Conclusion}
\label{sec:conclusion}
In this paper, we investigated the relationship between multimodal architecture design and modality gap dynamics. 
We formally characterized the geometric modality gap and theoretically showed that native architectures with a unified vocabulary intrinsically eliminate the initial modality gap, ensuring seamless geometric unity across all Transformer layers.
Empirically, comprehensive DNA-text evaluations confirm that our One Tokenizer achieves seamless integration directly within the LLM's native latent geometry. By consistently outperforming modular baselines, it proves that a unified tokenization space is a fundamentally superior framework for deep biological reasoning.

\textbf{Limitations and future Work.}
This study utilizes discrete symbolic modalities (i.e., DNA and text) to examine the model's internal representation dynamics, effectively bypassing the tokenization noise inherent in continuous signals (e.g., vision and audio). A key future direction is to investigate whether the seamless integration of a unified vocabulary persists when handling the vector quantization artifacts of continuous modalities. Additionally, while our theoretical bounds rigorously characterize the modality gap, exploring the evolution of latent geometry under massive-scale pre-training regimes remains an exciting frontier.

\section*{Acknowledgments and Disclosure of Funding}
We thank the anonymous reviewers.

\bibliographystyle{plainnat}
\bibliography{neurips_2026_ref}


\appendix


\section{Technical appendices and supplementary material}


In the supplementary material, we provide additional details about benchmarks, competitors and more experimental results. 
\begin{itemize}
    \item \textbf{Section \ref{sec:proof}.} We provide proofs for Lemma.\ref{lemma:main_init_gap} and Theorem.\ref{thm:main_deep_gap}. 
    \item \textbf{Section \ref{subsec:experimental_details}}. We provide more experimental details, including datasets, LLM-as-a-judge protocol, and evaluation metrics for domain gap.  
    \item \textbf{Section \ref{subsec:experimental_results}}. We provide more experimental results, including the detailed performance on NT with Qwen3-1.7B and Qwen3-8B, visualization on NT, and one reasoning case study. 
    \item \textbf{Section \ref{subsec:broader_impact}}. We discuss the broader impact. 
\end{itemize}

\subsection{Proofs}
\label{sec:proof}

\setcounter{lemma}{0}
\begin{lemma}[Embedding-level gap]
\label{lemma:init_gap}
Let $V_A$ and $V_B$ denote the modality-specific vocabulary for modalities $A$ and $B$, respectively. The expected modality gaps at the initial embedding layer ($l=0$) of the modular architecture $\Delta_\text{Mod}^{(0)}(\mathcal{H}_A,\mathcal{H}_B)$, and the native architecture $\Delta_\text{Nat}^{(0)}(\mathcal{H}_A,\mathcal{H}_B)$, satisfy:
\[
    \mathbb{E}[\Delta_\text{Mod}^{(0)}(\mathcal{H}_A,\mathcal{H}_B)] > \lim_{|V_A|, |V_B| \to \infty} \mathbb{E}[\Delta_\text{Nat}^{(0)}(\mathcal{H}_A,\mathcal{H}_B)] = 0
\]
\end{lemma}
\begin{proof}
We evaluate the expected modality gap at the initial embedding layer $l=0$. Following standard initialization practices in large language models, we assume the initial continuous representations (token embeddings) are drawn from independent isotropic Gaussian distributions. 

Based on our general metric, we first evaluate the modular architecture. Let the embeddings for modalities $A$ and $B$ be initialized as $e_{A} \sim \mathcal{N}(\mu_A, \sigma_A^2 I_d)$ and $e_{B} \sim \mathcal{N}(\mu_B, \sigma_B^2 I_d)$. The centroids $c_A^{(0)}$ and $c_B^{(0)}$, defined as the sample means of these embeddings, therefore satisfy:
\[
    c_A^{(0)} \sim \mathcal{N}\left(\mu_A, \frac{\sigma_A^2}{|V_A|} I_d\right), \quad c_B^{(0)} \sim \mathcal{N}\left(\mu_B, \frac{\sigma_B^2}{|V_B|} I_d\right)
\]
For the $d$-dimensional space, the expected squared $L_2$ norm between the two independent centroids is the sum of their variances and the squared distance between their expected means:
\begin{align*}
    \mathbb{E}[\|c_A^{(0)} - c_B^{(0)}\|_2^2] &= \sum_{i=1}^d \left[ \text{Var}(c_{A,i}^{(0)} - c_{B,i}^{(0)}) + (\mathbb{E}[c_{A,i}^{(0)} - c_{B,i}^{(0)}])^2 \right] \\
    &= d \left( \frac{\sigma_A^2}{|V_A|} + \frac{\sigma_B^2}{|V_B|} \right) + \|\mu_A - \mu_B\|_2^2
\end{align*}
Due to the concentration of measure in high-dimensional spaces ($d \gg 1$) \cite{vershynin2018high}, we approximate the expectation of the $L_2$ norm as $\mathbb{E}[\|c_A^{(0)} - c_B^{(0)}\|_2] \approx \sqrt{\mathbb{E}[\|c_A^{(0)} - c_B^{(0)}\|_2^2]}$. Substituting this into our gap metric yields:
\[
    \mathbb{E}[\Delta_\text{Mod}^{(0)}(\mathcal{H}_A, \mathcal{H}_B)] \approx \sqrt{ d \left( \frac{\sigma_A^2}{|V_A|} + \frac{\sigma_B^2}{|V_B|} \right) + \|\mu_A - \mu_B\|_2^2 } > 0
\]
Note that for any finite vocabularies $V_A$ and $V_B$, the sampling variance inherently prevents the gap from vanishing. Furthermore, because modular embeddings are parameterized in separate sub-spaces, they generally exhibit shifted means ($\mu_A \neq \mu_B$), making the initial gap structurally bounded away from zero.

Next, we evaluate the native architecture (e.g., One Tokenizer). With a shared embedding space, all disjoint tokens from $V_A$ and $V_B$ form a unified vocabulary $V_\text{Nat} = V_A \cup V_B$, with embeddings drawn from a single shared distribution $e_\text{Nat} \sim \mathcal{N}(\mu_\text{Nat}, \sigma_\text{Nat}^2 I_d)$. The centroids $c_A^{(0)}$ and $c_B^{(0)}$ are sample means from this identical space:
\[
    c_A^{(0)} \sim \mathcal{N}\left(\mu_\text{Nat}, \frac{\sigma_\text{Nat}^2}{|V_A|} I_d\right), \quad c_B^{(0)} \sim \mathcal{N}\left(\mu_\text{Nat}, \frac{\sigma_\text{Nat}^2}{|V_B|} I_d\right)
\]
The expected squared $L_2$ norm is now purely driven by sampling variance, as the mean difference $\|\mu_\text{Nat} - \mu_\text{Nat}\|_2^2$ perfectly cancels out:
\[
    \mathbb{E}[\|c_A^{(0)} - c_B^{(0)}\|_2^2] = d \cdot \sigma_\text{Nat}^2 \left( \frac{1}{|V_A|} + \frac{1}{|V_B|} \right)
\]
Substituting this into the general metric and taking the asymptotic limit as vocabulary sizes grow ($|V_A|, |V_B| \to \infty$), the sample means independently converge to the true population mean $\mu_\text{Nat}$, yielding:
\[
    \lim_{|V_A|, |V_B| \to \infty} \mathbb{E}[\Delta_\text{Nat}^{(0)}(\mathcal{H}_A, \mathcal{H}_B)] \approx \lim_{|V_A|, |V_B| \to \infty} \sqrt{d \sigma_\text{Nat}^2 \left( \frac{1}{|V_A|} + \frac{1}{|V_B|} \right)} = 0
\]
This mathematically demonstrates that while modular architectures maintain a strict positive gap $\mathbb{E}[\Delta_\text{Mod}^{(0)}] > 0$, the native architecture structurally forces the expected initial modality gap to vanish. This concludes the proof.
\end{proof}

\setcounter{theorem}{0}
\begin{theorem}[Layer-wise gap preservation]
\label{thm:deep_gap}
Given the initial condition established in Lemma \ref{lemma:main_init_gap}, and assuming the layer-wise transformation functions are Lipschitz continuous, the expected modality gaps $\Delta_\text{Mod}^{(l)}$ and $\Delta_\text{Nat}^{(l)}$ at every subsequent layer $l$ ($l > 0$) throughout the forward pass satisfy:
\[
    \mathbb{E}[\Delta_\text{Mod}^{(l)}(\mathcal{H}_A, \mathcal{H}_B)] > \lim_{|V_A|, |V_B| \to \infty} \mathbb{E}[\Delta_\text{Nat}^{(l)}(\mathcal{H}_A, \mathcal{H}_B)] = 0
\]
\end{theorem}

\begin{proof}
Let $f_\theta^{(l)}$ be the $l$-th Transformer block. While standard self-attention is not globally Lipschitz continuous \cite{kim2021lipschitz}, Layer Normalization strictly projects representations onto a bounded hypersphere \cite{ba2016layer}. Because this forms a compact domain, the continuously differentiable transformer block $f_\theta^{(l)}$ is mathematically guaranteed to satisfy local Lipschitz continuity with bound $B$:
\[
    \| c_A^{(l)} - c_B^{(l)} \|_2 = \| f_{\theta}^{(l)}(c_A^{(l-1)}) - f_{\theta}^{(l)}(c_B^{(l-1)}) \|_2 \leq B \| c_A^{(l-1)} - c_B^{(l-1)} \|_2
\]
By induction, the terminal gap is bounded by:
\[
    \| c_A^{(l)} - c_B^{(l)} \|_2 \leq B^l \| c_A^{(0)} - c_B^{(0)} \|_2
\]
Furthermore, deep Transformer blocks with residual connections are mathematically proven to act as bi-Lipschitz, information-preserving mappings that strictly prevent representation collapse \cite{dong2021attention, behrmann2019invertible}. Consequently, the distinct initial manifolds cannot trivially collapse into a single degenerate point. Thus, by applying Lemma \ref{lemma:init_gap}, the expected gap of the modular architecture at layer $l$ is lower-bounded away from zero and remains strictly positive:
\[
    0 < \mathbb{E} \left[ \Delta_{\text{Mod}}^{(l)}(\mathcal{H}_A, \mathcal{H}_B) \right] \leq B^l \cdot \mathbb{E}[\Delta_{\text{Mod}}^{(0)}(\mathcal{H}_A, \mathcal{H}_B)]
\]
For the native architecture, substituting the asymptotic result from Lemma \ref{lemma:init_gap}:
\[
    \lim_{|V_A|, |V_B| \to \infty} \mathbb{E} \left[ \Delta_{\text{Nat}}^{(l)}(\mathcal{H}_A, \mathcal{H}_B)\right] \leq B^l \cdot 0 = 0
\]
Since the distance metric is non-negative, the expected gap must be exactly zero. This concludes the proof.
\end{proof}

\subsection{More Experimental Details}
\label{subsec:experimental_details}
\subsubsection{Datasets}
\label{appendix:datasets}
\textbf{Nucleotide Transformer Benchmark (NT)} is a collection of genomic datasets containing a varied panel of 18 small-sized datasets and has been extensively evaluated in multiple studies of DNA foundation models. To make it suitable for DAN-language models, \cite{de2025multimodal} constructed an instruction-version, which contains a question-answer instruction of each DNA sequence. Specifically, they curated dozens of English questions and answers for each task and sampled a question-answer pair per input DNA sequence, where a special DNA token placeholder <DNA> is used when referring to the input DNA sequence and the  answer contains the classification label. We use the same train/test sets to evaluate all comparison models. 

\textbf{KEGG} is mainly designed to evaluate multi-step mechanistic reasoning \cite{fallahpour2025bioreason}. It contains 1,449 QA pairs, representing 173 variant networks and 37 unique diseases, where each question describes variant details, network definition, and gene description, and each answer provides concise mechanism-disease associations. We follow \cite{fallahpour2025bioreason} and use 1159/290 for train+validation/testing.

\begin{figure}[!htp]
    \centering
    \includegraphics[width=\linewidth]{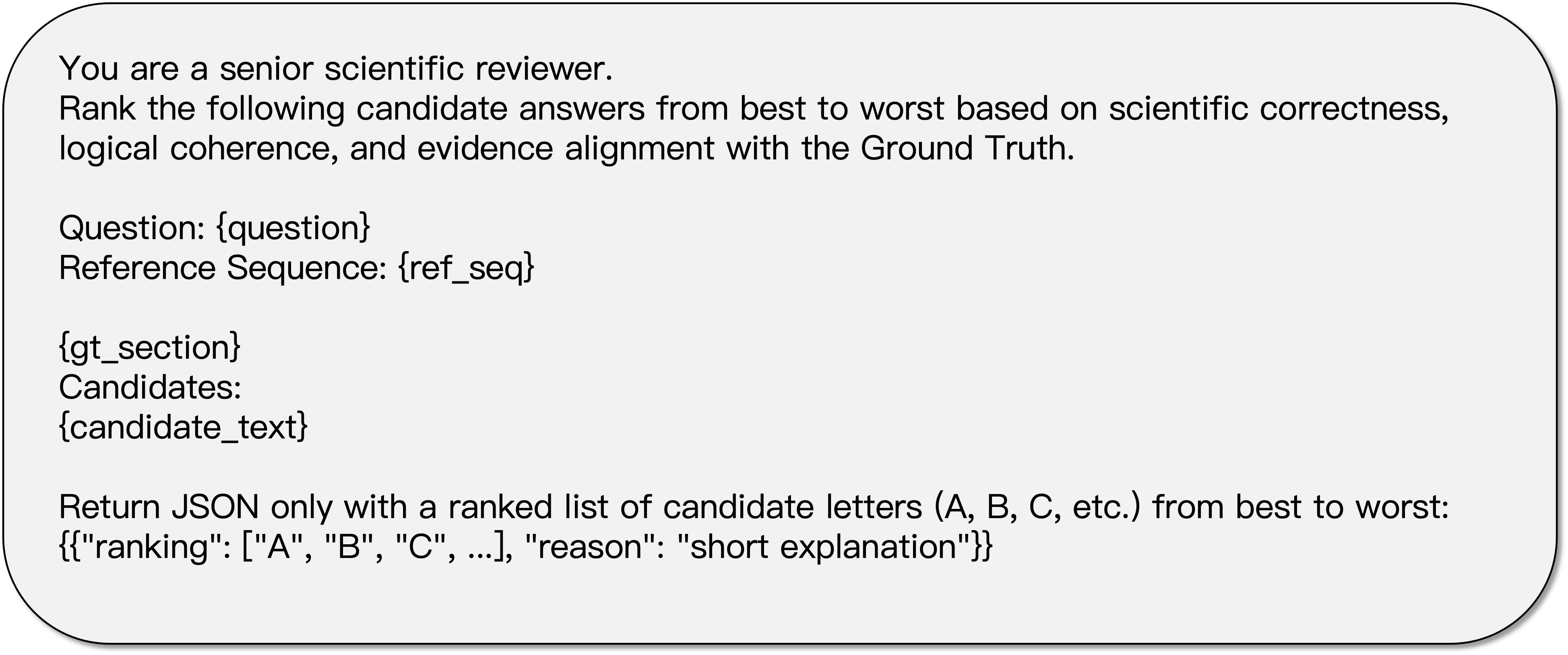}
    \caption{The prompt we used in our LLM-as-a-judge evaluation protocol.}
    \label{fig:llm-prompt}
\end{figure}

\subsubsection{LLM-as-a-Judge Protocol}
\label{appendix:llm}
We propose a novel LLM-as-a-judge protocol to quantitatively compare the reasoning performance of all methods inspired from \cite{chen2024mllm}. It assesses outputs based on relative ranking and aggregates judgments from a panel of $K$ distinct LLMs to mitigate evaluation bias. Specifically, for a given sample $x_i$, each LLM judge compares the reasoning outputs of all evaluated methods against the ground truth and provides a strict ranking, based on which we assign the ordinal score. Let $r_{i,k}(m)$ denote the rank assigned to method $m$ by judge $k$, where a lower value indicates a better rank. We convert the rank into scores using a predefined mapping $\mathrm{S}$. For example, if there are 5 outputs, scores could be $\mathrm{S}=\{5,4,3,2,1\}$. 
The score assigned by judge $k$ to method $m$ on sample $i$ becomes: $s_{i,k}(m)=\mathrm{S}(r_{i,k}(m))$. We then aggregate scores across $K$ judges and finally compute the overall performance of each method by averaging scores across the dataset: 
\begin{equation}
    s_{i}(m)=\sum_{k=1}^K s_{i,k}(m), S(m)=\frac{1}{N}\sum_{i=1}^N s_i(m)
\end{equation}
where $N$ denotes the number of samples. In our experiments, we use $K=3$ judge models - ``GPT-5.1'' ``Gemini-3.1-Pro-Preview'', and ``o3''- and assign scores $\{5,4,3,2,1\}$ based on ranking positions. The final score $S(m)$ reflects the relative reasoning quality of method $m$ over the dataset. This protocol ensures a fair, multi-perspective, and robust quantitative assessment of reasoning abilities. We show the detailed prompts in Fig.\ref{fig:llm-prompt}.

\subsubsection{Evaluation Metrics for Domain Gap}
\label{appendix:metrics}
\paragraph{Cohen Distance.}

To quantify the separation between DNA and text token embeddings while mitigating high-dimensional noise, we compute Cohen distance along the axis of maximal modality divergence within a denoised latent subspace. 

Let $\mathbf{X}^{(d)} \in \mathbb{R}^{n_d \times H}$ and $\mathbf{X}^{(t)} \in \mathbb{R}^{n_t \times H}$ denote the embedding matrices for DNA and text tokens, respectively. We first perform a joint Principal Component Analysis (PCA) on the concatenated embeddings $\mathbf{X} = [\mathbf{X}^{(d)}; \mathbf{X}^{(t)}]$ and project the tokens onto the first two principal components ($D=2$) to capture the dominant global geometry. Let $\hat{\mathbf{x}}_i \in \mathbb{R}^2$ denote the projected coordinates of a token in this joint PC1--PC2 plane.

We define the modality centroids, $\bar{\mathbf{x}}_{\text{DNA}}$ and $\bar{\mathbf{x}}_{\text{text}}$, as the empirical means of the projected tokens. The separation direction is defined as the unit vector $\mathbf{u}$ connecting these centroids:
\begin{equation}
\mathbf{u} = \frac{\bar{\mathbf{x}}_{\text{text}} - \bar{\mathbf{x}}_{\text{DNA}}}{\|\bar{\mathbf{x}}_{\text{text}} - \bar{\mathbf{x}}_{\text{DNA}}\|_2}.
\end{equation}
Each 2D token $\hat{\mathbf{x}}$ is then mapped to a univariate score $z = \hat{\mathbf{x}}^\top \mathbf{u}$ by projection onto $\mathbf{u}$, yielding two scalar distributions: $\mathcal{Z}_d = \{z^{(d)}_i\}_{i=1}^{n_d}$ and $\mathcal{Z}_t = \{z^{(t)}_j\}_{j=1}^{n_t}$.

The Cohen distance is calculated as the standardized mean difference:
\begin{equation}
d = \frac{\bar{z}_{\text{text}} - \bar{z}_{\text{DNA}}}{s_{\text{pooled}}}, \quad s_{\text{pooled}} = \sqrt{\frac{(n_d - 1)v_d + (n_t - 1)v_t}{n_d + n_t - 2}}
\end{equation}
where $\bar{z}$ and $v$ denote the sample mean and Bessel-corrected variance of the projected scores, respectively. This statistic measures the effect size of the modality gap along the most informative shared variance axes, providing a robust estimate of representation alignment.

\paragraph{KNN Fraction.}
After mapping DNA sequences and text sequences into a common low‑dimensional embedding space via PCA, we evaluate how well the two modalities intermix locally in PC1-PC2. For each data point $\mathbf{x}_i$ (which can be either a DNA sequence or a text document), we identify its $k$ nearest neighbors in the embedding space using Euclidean distance. Among these $k$ neighbors, we count how many belong to the \emph{opposite} modality (i.e., text neighbors for a DNA point, or DNA neighbors for a text point). The KNN fraction for point $\mathbf{x}_i$ is defined as 
\[
f_i = \frac{\text{ opposite-modality neighbours}}{k}
\]

A value $f_i \approx N_i/\sum_j{N}_j$ indicates that the point is predominantly surrounded by samples from the other modality, implying good cross‑modal alignment; a value $f_i \approx 0$ suggests strong modality‑specific clustering. The overall modality mix score of the hybrid model is given by harmonic mean of the per-modality mean KNN fractions $\bar{f}_{\text{DNA}}$ and $\bar{f}_{\text{text}}$. The harmonic mean is given by formula:
\[
H = \frac{2}{\frac{1}{\bar{f}_{\text{DNA}}} + \frac{1}{\bar{f}_{\text{text}}}} = \frac{2\,\bar{f}_{\text{DNA}}\,\bar{f}_{\text{text}}}{\bar{f}_{\text{DNA}}+\bar{f}_{\text{text}}}.
\]
This provides a single modality mix score that equally weights the mixing quality of both modalities.

\subsection{More Experimental Results}
\label{subsec:experimental_results}
\subsubsection{Comparison with State-of-the-Arts on NT}
\label{appendix:sota_nt}
Tab.\ref{tab:sota_nt} lists the detailed results of all comparing methods on NT, when using Qwen3-1.7B as the LLM backbone. Across models with 1.7B parameters, our method consistently ranks among the top-performing approaches, achieving either the best or second-best results across evaluated tasks. 

\begin{table}[!htbp]
    \centering
    \begin{tabular}{l c c c c}
        \hline
        \multirow{2}{*}{\textbf{Dataset}}  & \textbf{ChatNT*} & \textbf{BioReason}  & \textbf{SeqCLIP}&  \textbf{One Tokenizer}\\
        &(1.7B) & (1.7B)  &(1.7B) &  (1.7B)\\
        \hline
        \textit{Histone, and Enhancer} \\
        H3 histone  & 0.74 & 0.69  & 0.68& 0.76 \\
        H3K14ac  &0.81  &0.47  & 0.69  & 0.51 \\
        H3K36me3  & 0.10 & 0.46 & 0.51 & 0.58 \\
        H3K4me1 & 0.43 & 0.37 & 0.39  &0.47 \\
        H3K4me2 & 0.57  & 0.32  & 0.51  & 0.32 \\
        H3K4me3  & 0.66&0.34 & 0.56 &  0.34 \\
        H3K79me3 &0.18& 0.52 & 0.57 & 0.59 \\
        H3K9ac & 0.62& 0.51& 0.59 & 0.51 \\
        H4 histone& 0.76 & 0.77 & 0.75& 0.80 \\
        H4ac & 0.75 &0.43 & 0.68 & 0.46\\
        Enhancers & 0.04 &0.13 & 0.07 &0.49 \\
        Enhancer types & 0.00& -0.01& 0.02& 0.32\\
        \hline
        \textit{Promoter} \\
        Promoters  & 0.89 &0.88 &  0.89 & 0.87 \\
        Non-TATA promoters &0.90 & 0.88 & 0.89 & 0.88 \\
        TATA promoters & 0.90  & 0.85&  0.90& 0.86\\
        \hline
        \textit{Splice Site} \\
        Splice sites &  0.00 & 0.38  & 0.00 &0.72 \\
        Splice acceptors  & 0.76 & 0.02 &  0.56 & 0.88\\
        Splice donors  & 0.76 & 0.75  & 0.74& 0.82 \\
        \hline
        Average & 0.55 & 0.54& 0.56 & \textbf{0.62} \\
        \hline
    \end{tabular}
     \caption{Comparison with existing DNA-text models on NT.  Larger values indicate better performance. We follow the practice in ChatNT, training models on all tasks simultaneously and choosing MCC as the evaluation metric. $*$ means that we re-implement ChatNT using Qwen3-1.7B as the LLM backbone.}
    \label{tab:sota_nt}
\end{table}

\subsubsection{Visualization on NT}
\label{appendix:vis_nt}
Fig.\ref{fig:visualization_nt} illustrates the gene and text embeddings at 14-th layer on NT. Clearly, our One Tokenizer has a smaller gap than all modular competitors.
\begin{figure}[!htp]
    \centering
    \includegraphics[width=\linewidth]{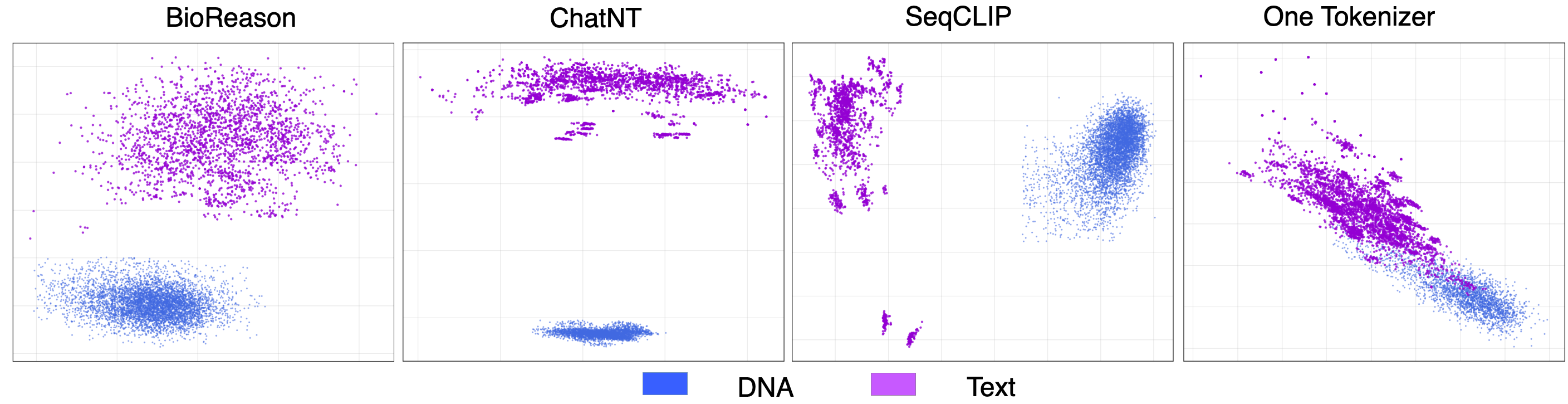}
    \caption{Visualizations of gene and text embeddings at 14-th layer on NT. Clearly, our One Tokenizer has a smaller gap than all modular competitors.}
    \label{fig:visualization_nt}
    \vspace{-5pt}
\end{figure}

\subsubsection{Reasoning Results}
\label{appendix:reasoning_results}
\begin{figure*}[!htp]
    \centering
    \includegraphics[width=\linewidth]{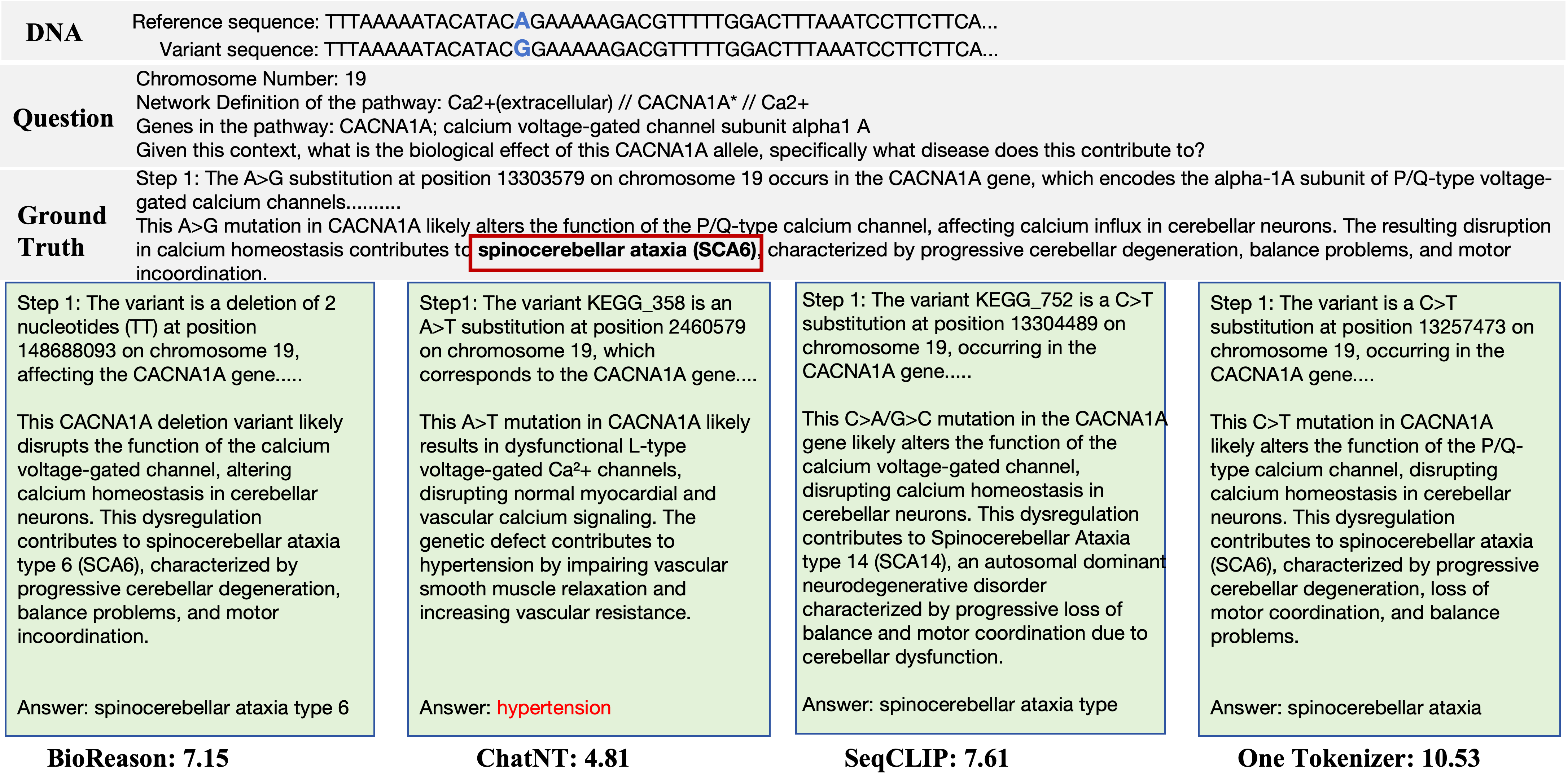}
    \caption{One KEGG reasoning case study. The red box indicates the ground truth disease; red text marks incorrect answers, and the numbers after each method denote the reasoning score.}
    \label{fig:qualitative_kegg}
\end{figure*}

Fig.\ref{fig:qualitative_kegg} presents an \emph{example} comparing the reasoning output using different methods querying about the biological effect of a CACNA1A allele on Chromosome 19, given its associated pathway information.  The following provides the detailed reasoning output of each method. 

It reveals different failure modes based on our evaluation criteria. 
BioReason identifies the correct disease, but its explanation mainly relies on gene–disease associations, with limited semantic detail and little mechanistic support. This results in lower scores in semantic similarity and completeness despite relatively fine logical consistency. 
ChatNT shows a more fundamental failure: it misidentifies the biological role of the target gene, which leads to domain-inconsistent reasoning across steps and ultimately generates an incorrect disease prediction, although the overall reasoning appears fluent. 
SeqCLIP produces long, step-by-step explanations with relatively coherent narratives; however, incorrect disease attribution and unsupported mechanistic assumptions significantly reduce its semantic alignment to the ground truth reasoning and its overall completeness. 
One Tokenzier operates directly on constructing biologically meaningful tokens, rather than relying on late-stage embedding alignment. Such an early-stage integration allows the model to follow coherent biological reasoning paths from molecular function to cellular context and disease phenotype, instead of handling these aspects as loosely related signals. Accordingly, One Tokenizer demonstrates strong and balanced performance on all three evaluation criteria.

\begin{itemize}
    \item \textbf{BioReason: }
        \begin{itemize}
            \item Step 1: The variant is a deletion of 2 nucleotides (TT) at position 148688093 on chromosome 19, affecting the CACNA1A gene.
            \item Step 2: CACNA1A encodes the alpha-1A subunit of the P/Q-type voltage-gated calcium channel, which is highly expressed in cerebellar Purkinje cells.
            \item Step 3: This deletion likely causes a frameshift or loss of a codon, potentially resulting in a truncated or dysfunctional calcium channel protein.
            \item Step 4: Dysfunctional P/Q-type calcium channels alter calcium influx into neurons, disrupting intracellular calcium homeostasis.

            \item Step 5: Altered calcium signaling affects neurotransmitter release at synapses, particularly in cerebellar circuits where CACNA1A is highly expressed.
            
            \item Step 6: Disrupted calcium homeostasis leads to Purkinje cell dysfunction and eventual neurodegeneration in the cerebellum.
            
            \item Step 7: The cerebellum is critical for motor coordination and balance, so cerebellar degeneration manifests as progressive ataxia.
            
            \item Step 8: This mutation pattern is consistent with spinocerebellar ataxia type 6 (SCA6), which is known to be caused by expansions of CAG trinucleotide repeats in CACNA1A.
            
            \item Step 9: The clinical presentation includes progressive loss of balance, coordination deficits, and eventual widespread cerebellar atrophy.
            
            \item Step 10: The disease follows an autosomal dominant inheritance pattern, with anticipation of expanded CAG repeats over time.
            
            \item This CACNA1A deletion variant likely disrupts the function of the calcium voltage-gated channel, altering calcium homeostasis in cerebellar neurons. This dysregulation contributes to spinocerebellar ataxia type 6 (SCA6), characterized by progressive cerebellar degeneration, balance problems, and motor incoordination. 
            
            Anwser:  spinocerebellar ataxia type 6 (SCA6)
        \end{itemize}

    \item \textbf{ChatNT:}
        \begin{itemize}
            \item Step 1: The variant KEGG-358 is an A$>$T substitution at position 2460579 on chromosome 19, which corresponds to the CACNA1A gene.
            
            \item  Step 2: This mutation likely alters the amino acid sequence of the calcium voltage-gated channel subunit alpha-1A protein (CACNA1A).
            
            \item Step 3: CACNA1A encodes a crucial component of L-type voltage-gated Ca²+ channels in cardiac and smooth muscle cells.
            
            \item Step 4: The mutation likely causes a point mutation that disrupts normal ion channel function or increases sensitivity to calcium influx regulation.
            
            \item Step 5: Dysregulation of extracellular calcium homeostasis via the CACNA1A pathway leads to abnormal intracellular calcium signaling.
            
            \item Step 6: Calcium dysregulation affects cellular processes including contraction, relaxation, and excitation in vascular smooth muscle and cardiomyocytes.
            
            \item Step 7: Disrupted calcium dynamics may cause impaired vasomotor tone control mechanisms in peripheral vessels.
            
            \item Step 8: Impaired endothelial NO production from arteriolar smooth muscle could contribute to reduced blood flow resistance.
            
            \item Step 9: Decreased vascular contractility with altered Ca²+ handling might lead to increased venous pressure and volume overload.
            
            \item Step 10: These physiological changes can manifest as chronic hypotension, edema, and pulmonary congestion.
            
            \item Step 11: The resulting cardiovascular instability suggests progressive development of hypertension and ischemic heart disease.
            
            \item This A>T mutation in CACNA1A likely results in dysfunctional L-type voltage-gated Ca²+ channels, disrupting normal myocardial and vascular calcium signaling. The genetic defect contributes to hypertension by impairing vascular smooth muscle relaxation and increasing vascular resistance. 
            
            Answer: hypertension 
        \end{itemize}

    \item \textbf{SeqCLIP: }
        \begin{itemize}
            \item Step 1: The variant KEGG-752 is a C$>$T substitution at position 13304489 on chromosome 19, occurring in the CACNA1A gene. 
            
            \item  Step 2: CACNA1A encodes the alpha-1A subunit of the P/Q-type voltage-gated calcium channel, which regulates calcium influx into neurons.
            
            \item Step 3: This mutation likely alters the structure or function of the calcium channel protein, potentially affecting its activation voltage, inactivation kinetics, or ion selectivity.
            
            \item Step 4: Altered calcium channel function disrupts intracellular calcium homeostasis, leading to abnormal calcium concentrations in the cytoplasm and endoplasmic reticulum.
            
            \item Step 5: Dysregulated calcium signaling affects multiple downstream cellular processes including neurotransmitter release, gene expression, and cell proliferation.
            
            \item Step 6: In cerebellar Purkinje cells, this calcium dysregulation particularly impacts synaptic transmission and neuronal excitability.
            
            \item Step 7: Chronic calcium imbalance triggers cellular stress responses, mitochondrial dysfunction, and impaired autophagy.
            
            \item Step 8: These cellular stresses lead to progressive Purkinje cell degeneration and eventual neurodegeneration in the cerebellum.
            
            \item Step 9: The cerebellum is critical for motor coordination and balance; loss of Purkinje cells results in spinocerebellar ataxia symptoms.
            
            \item Step 10: This specific mutation pattern is associated with Spinocerebellar Ataxia Type 14 (SCA14), an autosomal dominant form of the disease.
            
            \item This C>A/G>C mutation in the CACNA1A gene likely alters the function of the calcium voltage-gated channel, disrupting calcium homeostasis in cerebellar neurons. This dysregulation contributes to Spinocerebellar Ataxia type 14 (SCA14), an autosomal dominant neurodegenerative disorder characterized by progressive loss of balance and motor coordination due to cerebellar dysfunction.
            
            Answer: spinocerebellar ataxia type 14 (scna14) 
        \end{itemize}
    
    \item \textbf{One Tokenzier:}
        \begin{itemize}
            \item Step 1: The variant is a C>T substitution at position 13257473 on chromosome 19, occurring in the CACNA1A gene.
            \item Step 2: CACNA1A encodes the alpha-1A subunit of the P/Q-type voltage-gated calcium channel, which is highly expressed in cerebellar Purkinje cells and regulates calcium influx into neurons.
            \item Step 3: This mutation likely alters the protein structure or function of the calcium channel, potentially affecting voltage sensing, channel gating, or ion selectivity.
            \item Step 4: Altered calcium channel function disrupts intracellular calcium homeostasis, leading to dysregulated calcium influx into neurons.
            \item Step 5: Abnormal calcium signaling affects multiple downstream cellular processes including neurotransmitter release, gene expression, and cellular metabolism.
            \item Step 6: In cerebellar Purkinje cells, this calcium dysregulation leads to cellular stress and eventual neurodegeneration.
            \item Step 7: The cerebellum is critical for motor coordination, balance, and fine motor control.
            \item Step 8: Progressive cerebellar dysfunction results in the clinical manifestations of spinocerebellar ataxia.
            \item Step 9: The autosomal dominant inheritance pattern of this CACNA1A mutation means that a single mutated copy of the gene is sufficient to cause disease.
            \item Step 10: Patients develop progressive ataxia (loss of coordination), dysarthria (slurred speech), nystagmus (involuntary eye movements), and other cerebellar signs.
            \item Step 11: This mutation is associated with Spinocerebellar Ataxia type 6 (SCA6), one of the autosomal dominant spinocerebellar ataxias described in the disease association.
            \item This C>T mutation in CACNA1A likely alters the function of the P/Q-type calcium channel, disrupting calcium homeostasis in cerebellar neurons. This dysregulation contributes to spinocerebellar ataxia (SCA6), characterized by progressive cerebellar degeneration, loss of motor coordination, and balance problems.
            
            Anwser: spinocerebellar ataxia   

        \end{itemize}
\end{itemize}

\subsection{Broader Impact}
\label{subsec:broader_impact}
This paper studies the representation gap between modular and native multimodal large language models (MLLMs) from both theoretical and empirical perspectives. Our findings suggest that native architectures may facilitate more coherent cross-modal representations compared to modular designs. This insight could help guide the development of more efficient multimodal models that rely less on complex post-hoc alignment mechanisms, potentially improving training efficiency and reducing computational overhead. From a methodological perspective, we introduce a DNA-text setting as a controlled testbed for analyzing multimodal representation dynamics. While this setting is used for experimental validation, the underlying insights are not specific to biological data and may generalize to other discrete multimodal domains.

As with prior work on large language models, this research carries potential risks. Training large-scale models can contribute to environmental impact due to high computational requirements. In addition, models trained on large corpora may inherit and amplify biases present in the data. Although our experiments involve DNA sequences, the models presented in this work are designed for theoretical analysis rather than practical biological deployment. Nevertheless, we acknowledge that advances in generative modeling of biological sequences may raise dual-use concerns, such as the generation of inaccurate or misleading outputs in sensitive applications. To mitigate such risks, we emphasize that any real-world deployment of multimodal models in the life sciences should include appropriate safeguards, such as domain-specific validation, expert oversight, and, where applicable, experimental verification. More broadly, we advocate for responsible development practices, including transparency, evaluation on safety-critical tasks, and consideration of environmental efficiency. Overall, we hope this work encourages further research into multimodal architectures that are both effective and resource-efficient, while maintaining awareness of their broader societal implications.

\newpage
\section*{NeurIPS Paper Checklist}

\begin{enumerate}

\item {\bf Claims}
    \item[] Question: Do the main claims made in the abstract and introduction accurately reflect the paper's contributions and scope?
    \item[] Answer: \answerYes{} 
    \item[] Justification: Our paper analyzes the domain gap in multimodal large language models, from both theoretical and empirical perspectives. The abstract and introduction clearly state the claims and contributions. 
    \item[] Guidelines:
    \begin{itemize}
        \item The answer \answerNA{} means that the abstract and introduction do not include the claims made in the paper.
        \item The abstract and/or introduction should clearly state the claims made, including the contributions made in the paper and important assumptions and limitations. A \answerNo{} or \answerNA{} answer to this question will not be perceived well by the reviewers. 
        \item The claims made should match theoretical and experimental results, and reflect how much the results can be expected to generalize to other settings. 
        \item It is fine to include aspirational goals as motivation as long as it is clear that these goals are not attained by the paper. 
    \end{itemize}

\item {\bf Limitations}
    \item[] Question: Does the paper discuss the limitations of the work performed by the authors?
    \item[] Answer: \answerYes{}
    \item[] Justification: We discussed the limitations at the end of this paper in Sec.\ref{sec:conclusion}.
    \item[] Guidelines:
    \begin{itemize}
        \item The answer \answerNA{} means that the paper has no limitation while the answer \answerNo{} means that the paper has limitations, but those are not discussed in the paper. 
        \item The authors are encouraged to create a separate ``Limitations'' section in their paper.
        \item The paper should point out any strong assumptions and how robust the results are to violations of these assumptions (e.g., independence assumptions, noiseless settings, model well-specification, asymptotic approximations only holding locally). The authors should reflect on how these assumptions might be violated in practice and what the implications would be.
        \item The authors should reflect on the scope of the claims made, e.g., if the approach was only tested on a few datasets or with a few runs. In general, empirical results often depend on implicit assumptions, which should be articulated.
        \item The authors should reflect on the factors that influence the performance of the approach. For example, a facial recognition algorithm may perform poorly when image resolution is low or images are taken in low lighting. Or a speech-to-text system might not be used reliably to provide closed captions for online lectures because it fails to handle technical jargon.
        \item The authors should discuss the computational efficiency of the proposed algorithms and how they scale with dataset size.
        \item If applicable, the authors should discuss possible limitations of their approach to address problems of privacy and fairness.
        \item While the authors might fear that complete honesty about limitations might be used by reviewers as grounds for rejection, a worse outcome might be that reviewers discover limitations that aren't acknowledged in the paper. The authors should use their best judgment and recognize that individual actions in favor of transparency play an important role in developing norms that preserve the integrity of the community. Reviewers will be specifically instructed to not penalize honesty concerning limitations.
    \end{itemize}

\item {\bf Theory assumptions and proofs}
    \item[] Question: For each theoretical result, does the paper provide the full set of assumptions and a complete (and correct) proof?
    \item[] Answer: \answerYes{} 
    \item[] Justification: We provided the assumptions in the main paper and a complete proof in the appendix. 
    \item[] Guidelines:
    \begin{itemize}
        \item The answer \answerNA{} means that the paper does not include theoretical results. 
        \item All the theorems, formulas, and proofs in the paper should be numbered and cross-referenced.
        \item All assumptions should be clearly stated or referenced in the statement of any theorems.
        \item The proofs can either appear in the main paper or the supplemental material, but if they appear in the supplemental material, the authors are encouraged to provide a short proof sketch to provide intuition. 
        \item Inversely, any informal proof provided in the core of the paper should be complemented by formal proofs provided in appendix or supplemental material.
        \item Theorems and Lemmas that the proof relies upon should be properly referenced. 
    \end{itemize}

    \item {\bf Experimental result reproducibility}
    \item[] Question: Does the paper fully disclose all the information needed to reproduce the main experimental results of the paper to the extent that it affects the main claims and/or conclusions of the paper (regardless of whether the code and data are provided or not)?
    \item[] Answer: \answerYes{} 
    \item[] Justification: We provided the implementation details in Sec.\ref{sec:experiments}. 
    \item[] Guidelines:
    \begin{itemize}
        \item The answer \answerNA{} means that the paper does not include experiments.
        \item If the paper includes experiments, a \answerNo{} answer to this question will not be perceived well by the reviewers: Making the paper reproducible is important, regardless of whether the code and data are provided or not.
        \item If the contribution is a dataset and\slash or model, the authors should describe the steps taken to make their results reproducible or verifiable. 
        \item Depending on the contribution, reproducibility can be accomplished in various ways. For example, if the contribution is a novel architecture, describing the architecture fully might suffice, or if the contribution is a specific model and empirical evaluation, it may be necessary to either make it possible for others to replicate the model with the same dataset, or provide access to the model. In general. releasing code and data is often one good way to accomplish this, but reproducibility can also be provided via detailed instructions for how to replicate the results, access to a hosted model (e.g., in the case of a large language model), releasing of a model checkpoint, or other means that are appropriate to the research performed.
        \item While NeurIPS does not require releasing code, the conference does require all submissions to provide some reasonable avenue for reproducibility, which may depend on the nature of the contribution. For example
        \begin{enumerate}
            \item If the contribution is primarily a new algorithm, the paper should make it clear how to reproduce that algorithm.
            \item If the contribution is primarily a new model architecture, the paper should describe the architecture clearly and fully.
            \item If the contribution is a new model (e.g., a large language model), then there should either be a way to access this model for reproducing the results or a way to reproduce the model (e.g., with an open-source dataset or instructions for how to construct the dataset).
            \item We recognize that reproducibility may be tricky in some cases, in which case authors are welcome to describe the particular way they provide for reproducibility. In the case of closed-source models, it may be that access to the model is limited in some way (e.g., to registered users), but it should be possible for other researchers to have some path to reproducing or verifying the results.
        \end{enumerate}
    \end{itemize}

\item {\bf Open access to data and code}
    \item[] Question: Does the paper provide open access to the data and code, with sufficient instructions to faithfully reproduce the main experimental results, as described in supplemental material?
    \item[] Answer: \answerYes{} 
    \item[] Justification: The datasets are publicly available. We are constructing a website for this work and will release open-source code. 
    \item[] Guidelines:
    \begin{itemize}
        \item The answer \answerNA{} means that paper does not include experiments requiring code.
        \item Please see the NeurIPS code and data submission guidelines (\url{https://neurips.cc/public/guides/CodeSubmissionPolicy}) for more details.
        \item While we encourage the release of code and data, we understand that this might not be possible, so \answerNo{} is an acceptable answer. Papers cannot be rejected simply for not including code, unless this is central to the contribution (e.g., for a new open-source benchmark).
        \item The instructions should contain the exact command and environment needed to run to reproduce the results. See the NeurIPS code and data submission guidelines (\url{https://neurips.cc/public/guides/CodeSubmissionPolicy}) for more details.
        \item The authors should provide instructions on data access and preparation, including how to access the raw data, preprocessed data, intermediate data, and generated data, etc.
        \item The authors should provide scripts to reproduce all experimental results for the new proposed method and baselines. If only a subset of experiments are reproducible, they should state which ones are omitted from the script and why.
        \item At submission time, to preserve anonymity, the authors should release anonymized versions (if applicable).
        \item Providing as much information as possible in supplemental material (appended to the paper) is recommended, but including URLs to data and code is permitted.
    \end{itemize}

\item {\bf Experimental setting/details}
    \item[] Question: Does the paper specify all the training and test details (e.g., data splits, hyperparameters, how they were chosen, type of optimizer) necessary to understand the results?
    \item[] Answer: \answerYes{} 
    \item[] Justification: We provided all these details in the implementation details in Sec.\ref{sec:experiments}. 
    \item[] Guidelines:
    \begin{itemize}
        \item The answer \answerNA{} means that the paper does not include experiments.
        \item The experimental setting should be presented in the core of the paper to a level of detail that is necessary to appreciate the results and make sense of them.
        \item The full details can be provided either with the code, in appendix, or as supplemental material.
    \end{itemize}

\item {\bf Experiment statistical significance}
    \item[] Question: Does the paper report error bars suitably and correctly defined or other appropriate information about the statistical significance of the experiments?
    \item[] Answer: \answerNo{} 
    \item[] Justification: We follow current methods in the literature and report the numbers to indicate one method is better than the other. 
    \item[] Guidelines:
    \begin{itemize}
        \item The answer \answerNA{} means that the paper does not include experiments.
        \item The authors should answer \answerYes{} if the results are accompanied by error bars, confidence intervals, or statistical significance tests, at least for the experiments that support the main claims of the paper.
        \item The factors of variability that the error bars are capturing should be clearly stated (for example, train/test split, initialization, random drawing of some parameter, or overall run with given experimental conditions).
        \item The method for calculating the error bars should be explained (closed form formula, call to a library function, bootstrap, etc.)
        \item The assumptions made should be given (e.g., Normally distributed errors).
        \item It should be clear whether the error bar is the standard deviation or the standard error of the mean.
        \item It is OK to report 1-sigma error bars, but one should state it. The authors should preferably report a 2-sigma error bar than state that they have a 96\% CI, if the hypothesis of Normality of errors is not verified.
        \item For asymmetric distributions, the authors should be careful not to show in tables or figures symmetric error bars that would yield results that are out of range (e.g., negative error rates).
        \item If error bars are reported in tables or plots, the authors should explain in the text how they were calculated and reference the corresponding figures or tables in the text.
    \end{itemize}

\item {\bf Experiments compute resources}
    \item[] Question: For each experiment, does the paper provide sufficient information on the computer resources (type of compute workers, memory, time of execution) needed to reproduce the experiments?
    \item[] Answer: \answerYes{} 
    \item[] Justification: See the implementation details in Sec.\ref{sec:experiments}. We will release our source code. 
    \item[] Guidelines:
    \begin{itemize}
        \item The answer \answerNA{} means that the paper does not include experiments.
        \item The paper should indicate the type of compute workers CPU or GPU, internal cluster, or cloud provider, including relevant memory and storage.
        \item The paper should provide the amount of compute required for each of the individual experimental runs as well as estimate the total compute. 
        \item The paper should disclose whether the full research project required more compute than the experiments reported in the paper (e.g., preliminary or failed experiments that didn't make it into the paper). 
    \end{itemize}
    
\item {\bf Code of ethics}
    \item[] Question: Does the research conducted in the paper conform, in every respect, with the NeurIPS Code of Ethics \url{https://neurips.cc/public/EthicsGuidelines}?
    \item[] Answer: \answerYes{} 
    \item[] Justification: We reviewed the NeurIPS Code of Ethics carefully and conformed to it in every aspect. 
    \item[] Guidelines:
    \begin{itemize}
        \item The answer \answerNA{} means that the authors have not reviewed the NeurIPS Code of Ethics.
        \item If the authors answer \answerNo, they should explain the special circumstances that require a deviation from the Code of Ethics.
        \item The authors should make sure to preserve anonymity (e.g., if there is a special consideration due to laws or regulations in their jurisdiction).
    \end{itemize}

\item {\bf Broader impacts}
    \item[] Question: Does the paper discuss both potential positive societal impacts and negative societal impacts of the work performed?
    \item[] Answer: \answerYes{} 
    \item[] Justification: We discussed broader impacts in the appendix. 
    \item[] Guidelines:
    \begin{itemize}
        \item The answer \answerNA{} means that there is no societal impact of the work performed.
        \item If the authors answer \answerNA{} or \answerNo, they should explain why their work has no societal impact or why the paper does not address societal impact.
        \item Examples of negative societal impacts include potential malicious or unintended uses (e.g., disinformation, generating fake profiles, surveillance), fairness considerations (e.g., deployment of technologies that could make decisions that unfairly impact specific groups), privacy considerations, and security considerations.
        \item The conference expects that many papers will be foundational research and not tied to particular applications, let alone deployments. However, if there is a direct path to any negative applications, the authors should point it out. For example, it is legitimate to point out that an improvement in the quality of generative models could be used to generate Deepfakes for disinformation. On the other hand, it is not needed to point out that a generic algorithm for optimizing neural networks could enable people to train models that generate Deepfakes faster.
        \item The authors should consider possible harms that could arise when the technology is being used as intended and functioning correctly, harms that could arise when the technology is being used as intended but gives incorrect results, and harms following from (intentional or unintentional) misuse of the technology.
        \item If there are negative societal impacts, the authors could also discuss possible mitigation strategies (e.g., gated release of models, providing defenses in addition to attacks, mechanisms for monitoring misuse, mechanisms to monitor how a system learns from feedback over time, improving the efficiency and accessibility of ML).
    \end{itemize}
    
\item {\bf Safeguards}
    \item[] Question: Does the paper describe safeguards that have been put in place for responsible release of data or models that have a high risk for misuse (e.g., pre-trained language models, image generators, or scraped datasets)?
    \item[] Answer: \answerNA{}
    \item[] Justification: We do not envision misuse scenario by our model or data. 
    \item[] Guidelines:
    \begin{itemize}
        \item The answer \answerNA{} means that the paper poses no such risks.
        \item Released models that have a high risk for misuse or dual-use should be released with necessary safeguards to allow for controlled use of the model, for example by requiring that users adhere to usage guidelines or restrictions to access the model or implementing safety filters. 
        \item Datasets that have been scraped from the Internet could pose safety risks. The authors should describe how they avoided releasing unsafe images.
        \item We recognize that providing effective safeguards is challenging, and many papers do not require this, but we encourage authors to take this into account and make a best faith effort.
    \end{itemize}

\item {\bf Licenses for existing assets}
    \item[] Question: Are the creators or original owners of assets (e.g., code, data, models), used in the paper, properly credited and are the license and terms of use explicitly mentioned and properly respected?
    \item[] Answer: \answerYes{} 
    \item[] Justification: We used multiple open-sourced Github repositories that are free to use for non-commercial research purposes. 
    \item[] Guidelines:
    \begin{itemize}
        \item The answer \answerNA{} means that the paper does not use existing assets.
        \item The authors should cite the original paper that produced the code package or dataset.
        \item The authors should state which version of the asset is used and, if possible, include a URL.
        \item The name of the license (e.g., CC-BY 4.0) should be included for each asset.
        \item For scraped data from a particular source (e.g., website), the copyright and terms of service of that source should be provided.
        \item If assets are released, the license, copyright information, and terms of use in the package should be provided. For popular datasets, \url{paperswithcode.com/datasets} has curated licenses for some datasets. Their licensing guide can help determine the license of a dataset.
        \item For existing datasets that are re-packaged, both the original license and the license of the derived asset (if it has changed) should be provided.
        \item If this information is not available online, the authors are encouraged to reach out to the asset's creators.
    \end{itemize}

\item {\bf New assets}
    \item[] Question: Are new assets introduced in the paper well documented and is the documentation provided alongside the assets?
    \item[] Answer: \answerNo{} 
    \item[] Justification: We use publicly available benchmarks in experiments, and we do not introduce any new assets. 
    \item[] Guidelines:
    \begin{itemize}
        \item The answer \answerNA{} means that the paper does not release new assets.
        \item Researchers should communicate the details of the dataset\slash code\slash model as part of their submissions via structured templates. This includes details about training, license, limitations, etc. 
        \item The paper should discuss whether and how consent was obtained from people whose asset is used.
        \item At submission time, remember to anonymize your assets (if applicable). You can either create an anonymized URL or include an anonymized zip file.
    \end{itemize}

\item {\bf Crowdsourcing and research with human subjects}
    \item[] Question: For crowdsourcing experiments and research with human subjects, does the paper include the full text of instructions given to participants and screenshots, if applicable, as well as details about compensation (if any)? 
    \item[] Answer: \answerNA{} 
    \item[] Justification: Our paper does not involve crowdsourcing nor research with human subjects. 
    \item[] Guidelines:
    \begin{itemize}
        \item The answer \answerNA{} means that the paper does not involve crowdsourcing nor research with human subjects.
        \item Including this information in the supplemental material is fine, but if the main contribution of the paper involves human subjects, then as much detail as possible should be included in the main paper. 
        \item According to the NeurIPS Code of Ethics, workers involved in data collection, curation, or other labor should be paid at least the minimum wage in the country of the data collector. 
    \end{itemize}

\item {\bf Institutional review board (IRB) approvals or equivalent for research with human subjects}
    \item[] Question: Does the paper describe potential risks incurred by study participants, whether such risks were disclosed to the subjects, and whether Institutional Review Board (IRB) approvals (or an equivalent approval/review based on the requirements of your country or institution) were obtained?
    \item[] Answer:\answerNA{} 
    \item[] Justification: Our paper does not involve crowdsourcing nor research with human subjects.
    \item[] Guidelines:
    \begin{itemize}
        \item The answer \answerNA{} means that the paper does not involve crowdsourcing nor research with human subjects.
        \item Depending on the country in which research is conducted, IRB approval (or equivalent) may be required for any human subjects research. If you obtained IRB approval, you should clearly state this in the paper. 
        \item We recognize that the procedures for this may vary significantly between institutions and locations, and we expect authors to adhere to the NeurIPS Code of Ethics and the guidelines for their institution. 
        \item For initial submissions, do not include any information that would break anonymity (if applicable), such as the institution conducting the review.
    \end{itemize}

\item {\bf Declaration of LLM usage}
    \item[] Question: Does the paper describe the usage of LLMs if it is an important, original, or non-standard component of the core methods in this research? Note that if the LLM is used only for writing, editing, or formatting purposes and does \emph{not} impact the core methodology, scientific rigor, or originality of the research, declaration is not required.
    \item[] Answer: \answerYes{} 
    \item[] Justification: We used the LLM-as-a-Judge protocol in evaluating the reasoning ability of different methods, which are described in Sec.\ref{sec:experiments} in detail.
    \item[] Guidelines:
    \begin{itemize}
        \item The answer \answerNA{} means that the core method development in this research does not involve LLMs as any important, original, or non-standard components.
        \item Please refer to our LLM policy in the NeurIPS handbook for what should or should not be described.
    \end{itemize}

\end{enumerate}

\end{document}